\documentclass[xcolor=svgnames,11pt,reqno]{amsart}
\usepackage[round]{natbib}
\usepackage[T1]{fontenc}
\usepackage{amsmath}
\usepackage{amssymb}
\usepackage{amsthm}
\usepackage{fancyhdr}
\usepackage{enumitem}
\usepackage{comment}
\usepackage{relsize}
\usepackage{nicefrac}
\usepackage{bm}
\usepackage{bbm}
\usepackage{mathrsfs}
\usepackage{graphicx}
\usepackage{tabularx}
\usepackage[utf8]{inputenc}
\usepackage{cancel}
\usepackage{mathtools}
\usepackage{tikz}
\usepackage{dirtytalk,soul}
\usetikzlibrary{cd,decorations.pathreplacing,matrix,arrows,positioning,automata,shapes,shadows,calc,fadings,decorations,through,intersections}
\usepackage{float}
\usepackage{subcaption}
\usepackage{multirow}
\usepackage[left=2.5cm,top=2.8cm,right=2.5cm,bottom=2.5cm,head=3cm,headsep=1cm, foot=3cm]{geometry}
\usepackage{hyperref}
\usepackage[skip=2mm,indent=12pt]{parskip}
\definecolor{darkblue}{rgb}{0.1,0.1,0.9}
\definecolor{darkred}{rgb}{0.9,0.1,0.1}

\hypersetup{colorlinks, allcolors= darkblue}
\newtheorem{theorem}{Theorem}[section]
\newtheorem{definition}[theorem]{Definition}
\newtheorem{proposition}[theorem]{Proposition}
\newtheorem{assumption}[theorem]{Assumption}

\newtheorem{lemma}[theorem]{Lemma}

\newcommand{\cX}{\mathcal{X}}
\newcommand{\cL}{\mathcal{L}}
\newcommand{\cI}{\mathcal{I}}

\newcommand{\cN}{\mathcal{N}}
\newcommand{\cM}{\mathcal{M}}

\newcommand{\cC}{\mathcal{C}}

\renewcommand{\tilde}{\widetilde}

\newcommand{\E}{\mathbb{E}}
\renewcommand{\Pr}{\mathbb{P}}
\newcommand{\One}{\mathbbm{1}}
\newcommand{\R}{\mathbb{R}}

\renewcommand{\Pr}{\mathbb{P}}

\newcommand{\allocation}{(\{I_{ij}\}_{i\in\cN,j\in\cM},\{\pi_{ij}\}_{i\in\cN,j\in\cM})}
\newcommand{\allocationstar}{(\{I_{ij}^*\}_{i\in\cN,j\in\cM},\{\pi_{ij}^*\}_{i\in\cN,j\in\cM})}
\newcommand{\allocationtilde}{(\{\tilde{I}_{ij}\}_{i\in\cN,j\in\cM},\{\tilde{\pi}_{ij}\}_{i\in\cN,j\in\cM})}
\newcommand{\stratstar}{\big(\bm{\nu}_1^*,\ldots,\bm{\nu}_m^*,\mathfrak{I}_1^*,\ldots,\mathfrak{I}_n^*\big)}

\allowdisplaybreaks
\makeatletter

\newcommand{\Rmnum}[1]{\expandafter\@slowromancap\romannumeral #1@}

\title[Subgame Perfect Nash Equilibria in Large Reinsurance Markets]{\bf Subgame Perfect Nash Equilibria\vspace{0.25cm}\\In Large Reinsurance Markets\vspace{0.3cm}}

\author[Maria Andraos, Mario Ghossoub, and Michael B.\ Zhu]{Maria Andraos\vspace{0.1cm}\\ University of Waterloo\vspace{0.8cm}\\
 {Mario Ghossoub\vspace{0.1cm}\\ University of Waterloo\vspace{0.8cm}\\
  Michael B.\ Zhu\vspace{0.1cm}\\University of Waterloo\vspace{0.8cm}\\\today}}

\address{{\bf Maria Andraos}: University of Waterloo -- Department of Statistics and Actuarial Science -- 200 University Ave.\ W.\ -- Waterloo, ON, N2L 3G1 -- Canada\vspace{-0.25cm}}\email{\href{mailto:mandraos@uwaterloo.ca}{mandraos@uwaterloo.ca}\vspace{0cm}}

\address{{\bf Mario Ghossoub}: University of Waterloo -- Department of Statistics and Actuarial Science -- 200 University Ave.\ W.\ -- Waterloo, ON, N2L 3G1 -- Canada\vspace{-0.25cm}}
\email{\href{mailto:mario.ghossoub@uwaterloo.ca}{mario.ghossoub@uwaterloo.ca}\vspace{0cm}}

\address{{\bf Michael B.\ Zhu}:  University of Waterloo -- Department of Statistics and Actuarial Science -- 200 University Ave.\ W.\ -- Waterloo, ON, N2L 3G1 -- Canada\vspace{-0.25cm}} \email{\href{mailto:mbzhu@uwaterloo.ca}{mbzhu@uwaterloo.ca}\vspace{0cm}}

\thanks{\vspace{0.2cm}}

\thanks{ \vspace{0.2cm}}

\thanks{Mario Ghossoub acknowledges financial support from the Natural Sciences and Engineering Research Council of Canada (NSERC Grant No.\ 2024-03744). Michael B.\ Zhu acknowledges financial support from the Society of Actuaries through the Hickman Scholars Program.}

\thanks{}

\begin{document}
\sloppy

\maketitle

\begin{abstract}
We consider a model of a reinsurance market consisting of multiple insurers on the demand side and multiple reinsurers on the supply side, thereby providing a unifying framework and extension of the recent literature on optimality and equilibria in reinsurance markets. Each insurer has preferences represented by a general Choquet risk measure and can purchase coverage from any or all reinsurers. Each reinsurer has preferences represented by a general Choquet risk measure and can provide coverage to any or all insurers. Pricing in this market is done via a nonlinear pricing rule given by a Choquet integral. We model the market as a sequential game in which the reinsurers have the first-move advantage. We characterize the Subgame Perfect Nash Equilibria in this market in some cases of interest, and we examine their Pareto efficiency. In addition, we consider two special cases of our model that correspond to existing models in the related literature, and we show how our findings extend these previous results. Finally, we illustrate our results in a numerical example.
\end{abstract}

\bigskip
\section{Introduction}
\label{Intro}

Equilibrium models of risk-sharing markets in general, and (re)insurance markets in particular, were first introduced by \cite{Borch1962a} and \cite{Wilson1968}, using the framework of general equilibrium theory under uncertainty laid out by \cite{allais53b}, \cite{Arrow1953}, and \cite{debreu59}. This has set the paradigm for the field, based on the view that risks can be exchanged on markets for state-contingent goods, thereby relating price-mediated risk sharing to welfare economics, and notably to individually-rational Pareto-efficient allocations of risk. Along these lines, optimality in risk-exchange economies was examined through the prism of Pareto-efficiency in pure exchange economies with a state-contingent aggregate endowment. In particular, this led to a large literature on individually-rational Pareto-optimal (re)insurance design. We refer, for instance, to \cite{LoubergeDionne2025} and \cite{Gollier1992,Gollier2000,Gollier2013} for an overview. 

Concerns were then raised that individual rationality does not capture the reality that subsets of agents in the market may choose to deviate and form a separate risk-sharing market. This led to a cooperative game theoretic approach to insurance markets (e.g., \cite{Lemaire1977,Lemaire1991}), based on the stronger notion of collective rationality, whereby optimality of contractual agreements was then seen through the lens of various solution concepts from cooperative game theory, such as the core (e.g., \cite{BatonLemaire1981b})  or the bargaining set (e.g., \cite{BatonLemaire1981a}). Another discomfort with the general equilibrium approach to risk-sharing markets stemmed from the lack of considerations for strategic interactions in the classical Arrow-Debreu model. This prompted another approach to optimality of risk exchanges based either on sequential-move or simultaneous-move noncooperative game theoretic equilibrium concepts. The former category arises naturally in monopolistic insurance markets, and it includes optimality criteria based on the Stackelberg equilibrium. This was first examined in an insurance setting by \cite{ChanGerber}, in a market with one insurer and one reinsurer, both with expected-utility (EU) preferences. The authors gave a characterization of Stackelberg equilibria in the special case of exponential utilities. \cite{taylor1992}  extends these findings to more general risk exchanges with EU preferences, and \cite{CHEUNG201964} examine a two-agent market beyond EU, where preferences are represented by convex distortion risk measures (DRMs). \cite{li2021bowley} characterize Stackelberg equilibria in a two-agent market where preferences and premia use a mean-variance functional. In a dynamic two-market setting, similar studies have been proposed by \cite{chen2018new} and \cite{cao2022stackelberg}, for instance. \cite{boonenghossoub} provided the first study of the relationship between Stackelberg-equilibrium contracts and Pareto-optimal contracts, under fairly general preferences. The first extensions of the sequential-move (re)insurance market beyond the case of two agents were done by \cite{zhu2023equilibria} and \cite{ghossoub2024stackelberg}. The former extend the sequential-move market to account for multiple insurance providers who compete on pricing. They show that the Subgame-Perfect Nash Equilibrium (SPNE) is the appropriate equilibrium concept that extends the Stackelberg equilibrium to that case, and they characterize some examples of SPNEs, in the case where all market participants have preferences given by DRMs with heterogeneous beliefs and pricing is done via a Choquet functional. \cite{ghossoub2024stackelberg} provide a different extension of the two-agent case. Specifically, they maintain the assumption of monopoly insurance, but they consider the case of multiple policyholders. In their market, each policyholder evaluates risk according to a Choquet risk measure, and the monopolistic insurer uses a coherent risk measure. They characterize both Stackelberg equilibria and Pareto-efficient contracts.

In this paper, we provide a unifying framework, by examining a large market consisting of many insurers and many reinsurers, thereby extending the setting of both \cite{zhu2023equilibria} and \cite{ghossoub2024stackelberg}. Specifically, we consider a reinsurance market consisting of multiple insurers on the demand side and multiple reinsurers on the supply side. Each insurer has preferences represented by a general Choquet risk measure and can purchase coverage from any or all reinsurers. Each reinsurer has preferences represented by a general Choquet risk measure and can provide coverage to any or all insurers. Pricing in this market is done via a nonlinear pricing rule given by a Choquet integral, and we allow for belief heterogeneity. We model the market as a sequential-move game in which the reinsurers have the first-move advantage. They collectively observe the optimal demand from the policyholders as a function of prices. With that information in mind, they proceed to select the pricing vectors that minimize their risk exposure.

Our first result (Lemma \ref{SPNEminprop}) provides a way to determine the Subgame Perfect Nash Equilibria of this market game through backward induction. This allows for a characterization of the insurers' optimal strategies (optimal indemnity functions) as a best response to given reinsurers' strategies (pricing vectors) in terms of the optimal marginal indemnities for a given set of reinsurers' pricing rules. Furthermore, we give an explicit characterization of SPNEs in two important cases of interest: (i) the case where the reinsurers are all risk neutral (Theorem \ref{ThSPNE}), and (ii) the case where the initial insurable risks are comonotone (Theorem \ref{ThSPNEcom}). The former case is a standard assumption in the literature, and the latter case is meant to account for a worst-case dependence between the insurable losses in the market. Furthermore, we characterize the set of all Pareto-optimal reinsurance contracts (Proposition \ref{prop:optimal_indemnities}), and we show how the previously characterized SPNEs lead to Pareto-optimal contracts in this market (Theorem \ref{thm:induce_IRPO}).

Finally, to illustrate our findings, we consider a numerical example that relates to an important shortcoming of Stackelberg equilibria, initially documented by \cite{boonenghossoub}. In their two-agent market model, Stackelberg equilibria induces Pareto-efficient contracts that make the insurance buyer indifferent with the status quo, thereby transferring all consumer surplus to the monopolistic insurance provider. As subsequently noted by \cite{zhu2023equilibria}, the introduction of additional insurance providers to this market with one insurance buyer alleviates this problem, as price competition is indeed ultimately beneficial to the consumer. We show how the same principle applies to markets with multiple insurance buyers. Specifically, \cite{ghossoub2024stackelberg} consider a monopolistic reinsurance market with multiple insurance buyers, and they show that Stackelberg equilibria suffer from the same problem. Namely, the equilibrium contracts are Pareto optimal but lead to all consumer surplus being absorbed by the monopoly. We show how, with multiple insurers in a reinsurance market, introducing competition on the insurance supply side through price competition leads to a strict welfare gain for the insurance buyers at equilibrium.   

The rest of this paper is organized as follows. Section \ref{sec:formulation} presents the reinsurance market considered in this paper and formulates the key concepts used thereafter, including that of the SPNE. In Section \ref{sec:SPNE_characterization}, we provide a characterization of SPNEs in terms of backward induction, we describe the optimal indemnities given pricing vectors in terms of marginal compensation, and we give a characterization of SPNE in two cases of interest: the case of risk-neutral reinsurers and the case of comonotone losses. Additionally, Section \ref{SecPO} focuses on Pareto-efficient contracts in our market, and examines the efficiency of contracts resulting from SPNEs in the two important cases mentioned above. Section \ref{sec:example} illustrates the results of this paper in a numerical example, and it shows the effect of competition between reinsurers on the welfare of agents compared to the case of a market with a monopolistic reinsurer. Section \ref{sec:conclusion} concludes. Most of the proofs are presented in Appendix \ref{appendix-proofs}.

\bigskip 
\section{Problem Formulation}
\label{sec:formulation}

Let $\mathcal{X}$ denote the collection of all bounded measurable real-valued functions on a given measurable space $(S, \Sigma)$. We consider a reinsurance market populated by $n$ insurers, each subject to a random initial loss represented by a nonnegative random variable $X_i\in\mathcal{X}$, for $i\in\{1,\ldots,n\}$.
These insurers seek reinsurance from $m$ reinsurers present in the market.
For notational convenience, let $\mathcal{N} := \{1,\ldots,n\}$ denote the set of insurers indexed by $i$, and let $\cM := \{1,\ldots,m\}$ denote the set of reinsurers, indexed by $j$. Since each initial loss is bounded, we may define $M<\infty$ to be a uniform upper bound for all $X_i$.

\medskip
\subsection{Reinsurance Contracts}

In this market, each insurer $i \in \cN$ may cede a portion of their initial risk to reinsurer $j\in\cM$ by paying a premium.
Specifically, let $I_{ij}$ denote the indemnity function purchased by insurer $i$ from reinsurer $j$,
    and let $\pi_{ij}$ be the premium paid for this contract.
Then the risk exposure of each insurer $i \in \mathcal{N}$ after reinsurance is
    \begin{equation*}
         X_i - \sum_{j=1}^mI_{ij}(X_i) +\sum_{j=1}^m\pi_{ij}\,.
    \end{equation*}
The risk exposure of each reinsurer $j \in \cM$ is
\begin{equation*}
    \sum_{i=1}^n \bigg((1+\theta_j) I_{ij}(X_i) - \pi_{ij}\bigg)\,,
\end{equation*}
where $\theta_j>0$ is a loading factor that represents the cost of reinsurance to reinsurer $j$.

Denote the set of 1-Lipschitz and non-decreasing functions on $[0,M$] by
    \[
        \mathcal{I}_0=\left\{f:[0,M]\rightarrow[0,M]\quad\middle|\quad f(0)=0, \,\,\, 0 \leq f(x) - f(y) \leq x-y \,\,\, \forall\, 0\le y\le x\le M\right\}\,.
    \]

\noindent We assume throughout that only those indemnity functions that satisfy the no-sabotage condition of \cite{carlierdana2003} are admissible. Moreover, we assume that for each insurer $i\in\cN$, the sum of contracts also satisfies the no-sabotage condition, as given by the following:

\medskip

\begin{assumption}
    \label{as:sum_1-lipschitz}
    \[
        \sum_{j=1}^mI_{ij}\in\cI_0\,.
    \]
\end{assumption}

\noindent This assumption guarantees that each insurer does not purchase multiple reinsurance contracts for the same risk.
We refer to Appendix \ref{apdx:sum_1-lipschitz} for a more detailed justification of Assumption \ref{as:sum_1-lipschitz}. 
Hence, the set of admissible indemnity structures in this market is given by
    \[
        \mathcal{I}:=\left\{\{I_{ij}\}_{i\in \cN, j \in \cM}\,\,\middle|\,\,I_{ij} \in \mathcal{I}_0\quad\forall\,(i,j) \in \cN\times\cM\,, \,\,\, \text{and} \,\,\,   \sum_{j=1}^m I_{ij} \in \mathcal{I}_0\quad\forall i \in \cN \right\}\,. 
    \]

\noindent For notational convenience, we will define the set of admissible indemnity structures for each insurer $i$ by
    \[
        \mathcal{I}_i:=\left\{\{I_j\}_{j \in \cM}\,\,\middle|\,\,I_j\in \mathcal{I}_0\quad\forall\,j\in\cM\,, \,\,\, \text{and} \,\,\,   \sum_{j=1}^m I_j \in \cI_0\right\}\,. 
    \]

\noindent Note that $\prod_{i=1}^n\cI_i=\cI$.
The reinsurance contracts in aggregate completely determine how risk is reallocated among the agents through reinsurance.

\medskip

\begin{definition}[Allocation]
An \emph{allocation} is a pair
        \[
            \allocation\in\cI\times\mathbb{R}^{n\times m}\,,
        \]
    where $\{I_{ij}\}_{i \in \cN, j \in \cM}$ is an admissible vector of indemnity functions and $\{\pi_{ij}\}_{i\in \cN, j \in \cM}$ is a vector of premia.
\end{definition}

\medskip
\subsection{Risk Measures}

The preferences of each agent are represented by risk measures. We recall some standard properties of risk measures below.

\medskip

\begin{definition}
    A risk measure $\rho:\mathcal{X} \rightarrow \mathbb{R}$ is said to be:
    \begin{itemize}
        \item Monotone: if $\rho(X) \leq \rho(Y)$ whenever $X \leq Y$, for all $X, Y \in \mathcal{X}$.
        \item Positive Homogeneous: if $\rho(\lambda X) = \lambda \rho(X)$, for all $X \in \mathcal{X}$ and $\lambda \in \mathbb{R}^+$.
        \item Translation  invariant: if $\rho(X+c) = \rho(X) + c$, for all $X \in \mathcal{X}$ and $c \in \mathbb{R}$.
        \item Comonotone additive: if $\rho(X + Y) = \rho(X) + \rho(Y)$, for all $X,Y \in \cX$ that are comonotone.\footnote{Two random variables $X, Y \in \mathcal{X}$ are said to be comonotone if $(X(s_1) - X(s_2))(Y(s_1) -Y(s_2)) \geq 0$, for all $s_1,s_2\in S$.}
        \item Subadditive: if $\rho(X+Y) \leq \rho(X) + \rho(Y)$, for all $X, Y \in \mathcal{X}$.
    \end{itemize}
\end{definition}

We assume throughout that all risk measures are Choquet risk measures, as given by the following definitions.

\medskip

\begin{definition}
    A set function $\nu:\Sigma\rightarrow\R_+$ is a \emph{capacity} if
    \begin{itemize}
        \item $\nu(\varnothing) = 0$ and $\nu(S)<\infty$, and,
        \item If $A,B\in \Sigma$ are such that $A \subseteq B $, then $\nu(A) \leq \nu(B)$.
    \end{itemize}
    \end{definition}
    
\noindent Note that we require all capacities to be finite. Let $\mathcal{C}$ denote the set of all capacities.

\medskip

\begin{definition}
    Given a capacity $\nu \in \mathcal{C}$ and a random variable $X \in \mathcal{X}$, the \emph{Choquet expectation} of $X$ with respect to $\nu$ is defined as
        \begin{equation*}
            \int X\,d\nu := \int_0^{\infty}\nu(X>t)\,dt + \int_{-\infty}^0(\nu(X>t) -\nu(S))\,dt.
        \end{equation*}
\end{definition}

We denote the risk measure of insurer $i \in \mathcal{N}$ by $\rho_i^{In}$, and the risk measure of reinsurer $j\in\cM$ by $\rho_j^{Re}$.
For all risks $X\in\cX$, we assume that
$$\rho_i^{In}(X):=\int X\,d\alpha_i 
\ \ \hbox{and} \ \ 
\rho_j^{Re}(X):=\int X\,d\tau_j,$$
where $\alpha_i,\tau_j\in\cC$ for all $i\in\cN$ and $j\in\cM$.

\medskip

Note that Choquet risk measures are translation invariant, positively homogeneous, and comonotone additive.
Hence, the risk level of insurer $i\in\cN$ given the allocation $(\{I_{ij}\}_{i \in \cN, j \in \cM}, \{\pi_{ij}\}_{i\in \cN, j \in \cM})$ is
    \[
       \rho^{In}_i \left(X_i- \sum_{j=1}^mI_{ij}(X_i) +\sum_{j=1}^m\pi_{ij}\right)= \rho^{In}_i(X_i) -  \rho^{In}_i\left(\sum_{j=1}^mI_{ij}(X_i)\right) +\sum_{j=1}^m\pi_{ij}\,.
    \]

\noindent Similarly, the risk level of reinsurer $j\in\cM$ is
    \[
       \rho^{Re}_j\left(\sum_{i=1}^n \bigg((1+\theta_j) I_{ij}(X_i) - \pi_{ij}\bigg)\right)= (1+\theta_j) \rho_j^{Re}\left( \sum_{i=1}^n I_{ij}(X_i)\right) -\sum_{i=1}^n \pi_{ij}\,.
    \]

\medskip
\subsection{Market Structure and Equilibria}
\label{sec:equilibria}

We model the reinsurance market with a sequential structure, in which the reinsurers as a collective have the first-mover advantage.
That is, each reinsurer may determine the premium principle they wish to charge to each insurer before the reinsurance contracts are purchased.
For each $(i,j)\in\cN\times\cM$, let $\Pi_{ij}:\mathcal{X}\to\R$ denote the premium principle charged by reinsurer $j$ to insurer $i$.
We assume that each reinsurer uses a Choquet premium principle, of the form
    \[
        \Pi_{ij}(X)=\int X\,d\nu_{ij}, \ \forall \, X \in \cX,
    \]
where $\nu_{ij}\in\cC$.
The strategy of a reinsurer is then completely determined by a choice of pricing capacities, as made precise by the following.

\medskip

\begin{definition}
    For each $j\in\cM$, the strategy of reinsurer $j$ is a vector $\bm{\nu_j}$ of capacities $\nu_{ij}$ of length $n$:
        \[
            \bm{\nu}_j:=\begin{bmatrix}
                    \nu_{1j}\\
                    \vdots\\
                    \nu_{nj}
                \end{bmatrix}\in\mathcal{C}^n\,,
        \]

\noindent where $\nu_{ij}$ is the pricing capacity associated with the premium charged to insurer $i$ by reinsurer $j$.
\end{definition}

The strategy of each insurer $i \in \cN$ is to select an admissible indemnity structure from the set $\cI_i$, given the prices charged by each reinsurer $j \in \cM$.

\medskip

\begin{definition}
    The strategy of insurer $i \in \cN$ is a map $\mathfrak{I}_i:\cC^m\to\cI_i$. Specifically, for an $m$-dimensional vector of capacities $(\nu_{i1},\ldots,\nu_{im})$, the selection $\mathfrak{I}_i(\nu_{i1},\ldots,\nu_{im})$ represents the indemnities selected by insurer $i$ after observing the prices $(\nu_{i1},\ldots,\nu_{im})$. For this reason, we will also refer to maps of the form $\mathfrak{I}_i:\cC^m\to\cI_i$ as indemnity selections.
\end{definition}

We denote by $\mathfrak{I}_{ij}(\nu_{i1},\ldots,\nu_{im})$ the $j$-th entry of $\mathfrak{I}_i(\nu_{i1},\ldots,\nu_{im})$. This is the indemnity purchased by the $i$-th insurer from the $j$-th reinsurer under the strategy $\mathfrak{I}_i$, given that the pricing capacities of the reinsurers are $(\nu_{i1},\ldots,\nu_{im})$.

\medskip

\begin{definition}
    A \emph{strategy} is a tuple $(\bm{\nu}_1,\ldots,\bm{\nu}_m,\mathfrak{I}_1,\ldots,\mathfrak{I}_n)$ where:
        \begin{itemize}
            \item Each $\bm{\nu}_j$ is an $n$-dimensional vector of capacities representing reinsurer $j$'s strategy vector, for each $j \in \cM$.
            \item For each $i\in\mathcal{N}$, each $\mathfrak{I}_i$ is an indemnity selection. That is, $\mathfrak{I}_i$ is a map from $\mathcal{C}^m$ to $\mathcal{I}_i$.
        \end{itemize}
\end{definition}

The risk level to insurer $i \in \cN$ resulting from a strategy $(\bm{\nu}_1,\ldots,\bm{\nu}_m,\mathfrak{I}_1,\ldots,\mathfrak{I}_n)$  is therefore
    \[
        \rho_i^{In}(X_i)
            -\sum_{j=1}^m\rho_i^{In}\Big(\mathfrak{I}_{ij}(\nu_{i1},\ldots,\nu_{im})(X_i)\Big)
            +\sum_{i=1}^n\Pi_{ij}\Big(\mathfrak{I}_{ij}(\nu_{i1},\ldots,\nu_{im})(X_i)\Big)\,.
    \]

\noindent Similarly, the risk level to reinsurer $j \in \cM$ is
    \[
        \rho_j^{Re}\left(\sum_{i=1}^n\Big((1+\theta_j) \mathfrak{I}_{ij}(\nu_{i1},\ldots,\nu_{im})(X_i)\Big)\right)
            -\sum_{i=1}^n\Pi_{ij}\Big(\mathfrak{I}_{ij}(\nu_{i1},\ldots,\nu_{im})(X_i)\Big)\,.
    \]

\medskip
\subsubsection{Subgame Perfect Nash Equilibria}

Our focus is to identify and characterize the Subgame Perfect Nash Equilibria (SPNE) within this economic game.

\medskip

\begin{definition}
\label{defn:ne}
    A \emph{Nash Equilibrium} (NE) is a strategy from which no player has an incentive to deviate. Precisely, a strategy $\stratstar$ is a Nash Equilibrium if the following hold:
    \begin{enumerate}[label=(\roman*)]
        \item \label{defn:ne_1}
            There does not exist a capacity vector $\bm{\tilde{\nu}}\in\mathcal{C}^n$ and a reinsurer $j\in\cM$ such that 
                \begin{align*}
                    &\rho^{Re}_j \bigg(\sum_{i=1}^n (1+\theta_j)\mathfrak{I}^*_{ij}(\tilde{\nu}_i, \nu^*_{i,-j})(X_i)\bigg)
                        -\sum_{i=1}^n \Pi^{\tilde{\nu}_i}\bigg(\mathfrak{I}^*_{ij}(\tilde{\nu}_i, \nu^*_{i,-j})(X_i)\bigg)\\
                    &<\rho^{Re}_j \bigg(\sum_{i=1}^n (1+\theta_j) \mathfrak{I}^*_{ij}(\nu_{i1}^*, \ldots, \nu^*_{im})(X_i)\bigg)
                        -\sum_{i=1}^n\Pi^{\nu^*_{ij}}\bigg(\mathfrak{I}^*_{ij}(\nu^*_{i1},\ldots, \nu^*_{im})(X_i)\bigg)\,,
                \end{align*}
            where $\tilde{\bm{\nu}}$ is an $n$-dimensional column vector of capacities, and $(\tilde{\nu}_i,\nu_{i,-j})$ is the $m$-dimensional row vector $(\nu_{i1},\ldots,\nu_{im})$ with its $j$-th element replaced by $\tilde{\nu}_i$.
            In other words, reinsurer $j$ cannot improve by changing their strategy from $\bm{\nu}_j^*$ to $\tilde{\bm{\nu}}$, all else held constant.
        \item \label{defn:ne_2}
            There does not exist an insurer $i\in\cN$ and an indemnity selection $\tilde{\mathfrak{I}}_i$ such that
                \begin{align*}
                    &\rho^{In}_i(X_i)
                        -\rho^{In}_i\bigg(\sum_{j=1}^m \tilde{\mathfrak{I}}_{ij}(\nu^*_{i1},\ldots, \nu^*_{im})(X_i)\bigg)
                        +\sum_{j=1}^m \Pi^{\nu^*_{ij}} \bigg( \tilde{\mathfrak{I}}_{ij}(\nu^*_{i1},\ldots, \nu^*_{im})(X_i)\bigg)\\
                    &<\rho^{In}_i(X_i)
                        -\rho^{In}_i\left(\sum_{j=1}^m \mathfrak{I}^*_{ij}(\nu^*_{i1},\ldots, \nu^*_{im})(X_i)\right)
                        +\sum_{j=1}^m\Pi^{\nu^*_{ij}}\bigg(\mathfrak{I}^*_{ij}(\nu^*_{i1},\ldots, \nu^*_{im})(X_i)\bigg).
                \end{align*}
            That is, insurer $i$ does not improve by changing their indemnity selection from $\mathfrak{I}_i^*$ to $\tilde{\mathfrak{I}}_i$.
    \end{enumerate}
\end{definition}

\medskip

\begin{definition}
\label{defn:spne}
    A strategy $\stratstar$ is a Subgame Perfect Nash Equilibria (SPNE) if the following hold:
    \begin{enumerate}[label=(\roman*)]
        \item \label{defn:spne_1}
            $\stratstar$ is an NE.
        \item \label{defn:spne_2}
            For any choice of capacity vectors $(\bm{\nu}_1,\ldots,\bm{\nu}_m) \in \mathcal{C}^{n \times m}$, there does not exist an insurer $i\in\cN$ and an indemnity selection $\tilde{\mathfrak{I}}_i$ such that
            \begin{align*}
                &\rho^{In}_i(X_i) - \rho^{In}_i\bigg(\sum_{j=1}^m \tilde{\mathfrak{I}}_{ij}(\nu_{i1},\ldots, \nu_{im})(X_i)\bigg) +\sum_{j=1}^m \Pi^{\nu_{ij}} \bigg( \tilde{\mathfrak{I}}_{ij}(\nu_{i1},\ldots, \nu_{im})(X_i)\bigg) \\
            &< \rho^{In}_i(X_i) - \rho^{In}_i\bigg(\sum_{j=1}^m \mathfrak{I}^*_{ij}(\nu_{i1},\ldots, \nu_{im})(X_i)\bigg) +\sum_{j=1}^m \Pi^{\nu_{ij}}\bigg( \mathfrak{I}^*_{ij}(\nu_{i1},\ldots, \nu_{im})(X_i)\bigg)\,.
        \end{align*}
    \end{enumerate}
\end{definition}

Note that by definition, every SPNE is an NE, but the converse does not necessarily hold.
The main difference is that Definition \ref{defn:spne}\ref{defn:spne_2} must hold for \emph{any} choice of capacity vectors $(\bm{\nu}_1,\ldots,\bm{\nu}_m) \in \mathcal{C}^{n \times m}$, instead of only holding for the specific choice $(\bm{\nu}^*_1,\ldots,\bm{\nu}^*_m)$.
For an in-depth discussion of the motivation for this distinction, we refer to \cite[Chapter 7]{maschler2020game}, for instance.

\bigskip
\section{Characterization of SPNEs}
\label{sec:SPNE_characterization}

\subsection{Backward Induction}
Our main result in this section is a characterization of SPNEs in a market with multiple insurers and reinsurers.
It is well-known (e.g., \citealt[Remark 4.50]{maschler2020game}) that all SPNEs may be found through the process of backward induction.
That is, to determine an SPNE, we first determine the optimal strategy for each insurer in the market, since the insurers are the last agents to act.
Then, given one such optimal strategy for each insurer, we then determine the pricing capacities for the reinsurers that form an NE.
This process is illustrated in the context of the present paper by the following lemma.

\medskip

\begin{lemma} 
\label{SPNEminprop}
A strategy $\stratstar$ is an SPNE if and only if the following hold:
    \begin{enumerate}[label=(\roman*)]
        \item \label{SPNEminprop_1}
        For any choice of pricing capacities $(\bm{\nu}_1,\ldots,\bm{\nu}_m)$ $\in \mathcal{C}^{n \times m}$ and for all $i\in\cN$, the indemnity structure $\mathfrak{I}^*_i(\nu_{i1},\ldots, \nu_{im})$ solves 
            \begin{equation} \label{eq:minproblemSPNE}
                \min_{\{I_{ij}\}_{j\in \cM} \in \mathcal{I}_i} \left\{ \rho^{In}_i(X_i) - \sum_{j=1}^m \rho^{In}_i(I_{ij}(X_i)) + \sum_{j=1}^m \Pi^{\nu_{ij}}(I_{ij}(X_i)) \right\}.
            \end{equation}
        
        \item \label{SPNEminprop_2}
        $(\bm{\nu^*}_1,\ldots,\bm{\nu^*}_m)$ is a Nash Equilibrium for the reduced game formed from fixing the strategy $\mathfrak{I}^*_i$ of each insurer $i\in\cN$. That is, there does not exist $j \in \cM$ and $\tilde{\bm{\nu}}\in \mathcal{C}^n$ such that:
            \begin{align*}
                &\rho_j^{Re}\left(\sum_{i=1}^n (1+\theta_j) \mathfrak{I}^*_{ij}(\tilde{\nu}_i, \nu^*_{i,-j})(X_i)\right)
                - \sum_{i=1}^n \Pi^{\tilde{\nu}_i}\left(\mathfrak{I}^*_{ij}(\tilde{\nu}_i, \nu^*_{i,-j})(X_i) \right) \\ &\quad< \rho_j^{Re}\left(\sum_{i=1}^n (1+\theta_j)\mathfrak{I}^*_{ij}(\nu^*_{i1}, \cdots, \nu^*_{im})(X_i) \right) -  \sum_{i=1}^n \Pi^{\nu^*_{ij}}\left( \mathfrak{I}^*_{ij}(\nu^*_{i1}, \cdots, \nu^*_{im})(X_i)\right)\,.
                    \stepcounter{equation}\tag{\theequation}\label{eq:no_deviation}
            \end{align*}
    \end{enumerate}
\end{lemma}

\begin{proof}
    Note that point \ref{SPNEminprop_1} is a restatement of Definition \ref{defn:spne}\ref{defn:spne_2}.
    Hence, by Definition \ref{defn:spne}, it suffices to show that if \ref{SPNEminprop_1} holds, a strategy $\stratstar$ is an NE if and only if \ref{SPNEminprop_2} holds.
    To this end, if \ref{SPNEminprop_1} holds, it follows that for all $i\in\cN$, the indemnity structure $\mathfrak{I}^*_i(\nu_{i1}^*,\ldots,\nu_{im}^*)$ solves
        \[
            \min_{\{I_{ij}\}_{j\in \cM}\in\mathcal{I}_i}
                \left\{\rho^{In}_i(X_i)
                -\sum_{j=1}^m\rho^{In}_i(I_{ij}(X_i))
                +\sum_{j=1}^m\Pi^{\nu_{ij}^*}(I_{ij}(X_i))\right\}\,,
        \]
    which implies that Definition \ref{defn:ne}\ref{defn:ne_2} holds.
    Since \ref{SPNEminprop_2} is a restatement of Definition \ref{defn:ne}\ref{defn:ne_1}, the result follows.
\end{proof}

Lemma \ref{SPNEminprop} implies that the process of characterizing SPNE strategies involves first solving for insurer strategies before solving for reinsurer strategies.
To this end, we first determine all candidate insurer strategies by identifying the indemnity structures solving \eqref{eq:minproblemSPNE}.
This is given by the following result, which characterizes each indemnity function through its derivative using the so-called Marginal Indemnification Function approach.

\medskip

\begin{proposition}
    \label{propcapacity}
    Let $(\bm{\nu}_1,\ldots,\bm{\nu}_m)\in \mathcal{C}^{n \times m}$ be fixed pricing capacities, and define, for each $i\in\cN$ and $z\in[0,M]$,
        \begin{align*}
            \underline{\nu}_i(X_i>z)&:=\min_{j\in \cM} \nu_{ij}(X_i>z)\,,\\
            \mathcal{M}_i(z)&:=\{j\in\cM\,:\,\nu_{ij}(X_i>z)=\underline{\nu}_i(X_i>z)\}\,. 
        \end{align*}
    
    \noindent Then for a fixed $i\in\cN$, an indemnity structure $\{I^*_{ij}\}_{j \in \cM}\in\cI_i$ is a solution to \eqref{eq:minproblemSPNE} if and only if for each $j\in\cM$ and $x\in[0,M]$, the indemnity $I_{ij}^*$ can be written as
        \begin{equation}
        \label{eq:gamma_defn}
            I^*_{ij}(x) = \int_0^x \gamma^*_{ij}(z)\,dz\,,
        \end{equation}
    where each $\gamma_{ij}^*$ is a $[0,1]$-valued measurable function that satisfies, for almost all $z\in[0,M]$,
        \[
            \gamma^*_{ij}(z)=h_{ij}(z)\One_{\{j \in \mathcal{M}_i(z)\}},
        \]
    and  
        \[
            \sum_{j=1}^m\gamma^*_{ij}(z)=\begin{cases}
                1,&\mbox{ if}\quad\alpha_i(X_i >z)>\underline{\nu}_i(X_i>z)\\
                H_i(z),&\mbox{ if}\quad\alpha_i(X_i >z)=\underline{\nu}_i(X_i>z)\\
                0,&\mbox{ if}\quad\alpha_i(X_i >z)<\underline{\nu}_i(X_i>z)
            \end{cases}\,,
        \]
    for some $[0,1]$-valued measurable functions $H_i$ and $h_{ij}$.
\end{proposition}

\begin{proof}
    Note that if prices are fixed, the choice of each individual insurer $i\in\cN$ is not affected by the decisions of other insurers in the market.
    Hence, when $i\in\cN$ is fixed, this market reduces to the case with the one insurer $i$ and multiple reinsurers $j\in\cM$.
    It then follows that an application of \cite[Theorem 3.1]{BOONEN202123} to each $i\in\cN$ yields the result.
\end{proof}

Note that Proposition \ref{propcapacity} characterizes each indemnity function $I_{ij}^*$ using its derivative $\gamma_{ij}^*$ and the relationship given by \eqref{eq:gamma_defn}.
Hence, we refer to the $[0,1]$-valued function $\gamma_{ij}^*$ as the \emph{marginal indemnification} of $I_{ij}^*$. Furthermore, Proposition \ref{propcapacity} implies that the solutions to \eqref{eq:minproblemSPNE} are not unique, since multiple reinsurers may be charging the lowest price (i.e., the set $\cM_i(z)$ is not a singleton).
Different indemnity structures may also result from the flexibility when $\alpha_i=\underline{\nu}_i$ for an insurer $i$.
To facilitate an explicit characterization of SPNEs, we consider only those solutions to \eqref{eq:minproblemSPNE} that satisfy the property of \emph{generous distribution} introduced in \cite{zhu2023equilibria}.

\medskip

\begin{definition}[Generous Distribution]
\label{Generous Distribution}
    For a fixed $i\in\cN$, let $\{I_{ij}^*\}_{j\in\cM}$ be a solution to \eqref{eq:minproblemSPNE}.
    By \eqref{eq:gamma_defn}, we may write $I^*_{ij}(x)=\int_0^x\gamma^*_{ij}(z)\,dz$, where $\gamma^*_{ij}$ is a $[0,1]$-valued function.
    For each $z\in[0,M]$, let
        \begin{align*}
            \gamma^*_{i0}(z)&:=1-\sum_{j=1}^m\gamma^*_{ij}(z)\,,\\
            \tau_{i,0}(X_i>z)&:=\alpha_i(X_i>z)\,,\quad\mbox{and}\,,\\
            \tau_{i,j}(X_i>z)&:=(1+\theta_j)\tau_j(X_i>z)\quad\forall j\in\cM\,.
     \end{align*}
     
     \noindent Then $\{I_{ij}^*\}_{j\in\cM}$ \emph{distributes generously} if for almost all $z\in[0,1]$ and all $k \in\cM\cup\{0\}$, we have $\gamma^*_{ik}(z)=0$ if both of the following hold: 
        \begin{itemize}
            \item $\nu_{ik}(X_i>z) = \min_{j \in \cM \cup \{0\}} \{ \nu_{ij}(X_i>z)\}$, and,
            \item There exists $k'\in \big(\cM \cup \{0\}\big) \setminus \{k\}$ such that
                \[
                    \nu_{ik'}(X_i>z) = \nu_{ik}(X_i>z) \,\,\, \text{but}\,\,\, \tau_{i,k'}(X_i>z) < \tau_{i,k}(X_i>z)\,.
                \]
        \end{itemize}
\end{definition}

For the remainder of this section, we impose the following assumption.

\medskip

\begin{assumption}
    We assume that each insurer $i \in \mathcal{N}$ always selects an optimal indemnity that distributes generously.
    That is, for every fixed set of capacities $(\bm{\nu}_1,\ldots,\bm{\nu}_m)\in \mathcal{C}^{n \times m}$, the indemnity structure $\mathfrak{I}^*_i(\nu_{i1},\ldots,\nu_{im})$ distributes generously in the sense of Definition \ref{Generous Distribution}, for all $i\in\cN$.
    Let $\aleph$ denote the set of such strategies $(\mathfrak{I}_1^*,\ldots,\mathfrak{I}_n^*)$.
\end{assumption}

It now remains to determine the reinsurer strategies after fixing some insurer strategies $(\mathfrak{I}_1^*,\ldots,\mathfrak{I}_n^*)\in\aleph$.
These reinsurer strategies are determined constructively by using the following definitions.

\medskip

\begin{definition}
    For each insurer  $i \in \mathcal{N}$, we define the capacity $\Bar{\tau}_i$ to be the pointwise second-lowest function of the set $\{\tau_{i,0},\tau_{i,1}, \ldots, \tau_{i,m}\}$. We refer to $\Bar{\tau}_i$ as the second-lowest true preference for insurer $i$.
\end{definition}

\medskip

\begin{definition}
    \label{defn:beth}
    Define the class of reinsurer strategies $\beth \subseteq \mathcal{C}^{n\times m}$ as the set of all $(\bm{\nu}_1,\ldots,\bm{\nu}_m)$ such that for all $i\in\cN$ and almost all $z \in [0,M]$,
    \begin{enumerate}
        \item $\min_{j\in \cM} \nu_{ij}(X_i>z)=\Bar{\tau}_i(X_i>z)$.
            That is, the lowest price charged by any reinsurer to indemnify the tail risk $X_i>z$ is given by the second-lowest true preference $\bar{\tau}_i$.
        \item  There exists $k,k'\in\cM\cup\{0\}\,,k \neq k'$ such that
            \[
                \min_{j \in \cM}\nu_{ij}(X_i>z) = \nu_{ik}(X_i>z) = \nu_{ik'}(X_i>z)\,.
            \]
            That is, for each $i\in\cN$, there are always at least two reinsurers charging a price equal to the second-lowest true preference.
        \item We have $\mathcal{T}_i(z) \neq \varnothing \implies \mathcal{T}_i(z) \cap \mathcal{M}_i(z) \neq \varnothing$ , where
            \[
                \mathcal{T}_i(z):=\left\{j\in\cM\,\,\middle|\,\,\tau_{i,j}(X_i>z)  = \min_{k \in \cM \cup \{0\}} \tau_{i,k}(X_i>z)\right\}\,.
            \]
            In other words, for each $i\in\cN$, out of the reinsurers $j$ with the lowest true preference $\tau_{i,j}$, at least one is charging the price $\bar{\tau}_i$ to insurer $i$.           
    \end{enumerate}
\end{definition}

The main result of \cite{zhu2023equilibria} shows that in a model with only one policyholder, every strategy in $\beth\times\aleph$ is an SPNE.
However, since there are multiple policyholders in our model, each reinsurer must consider the risk of the sum of all contracts that they have written.
We show in the following that a characterization is possible under either of two additional assumptions: (i) that the reinsurers are risk neutral, or (ii) that the initial risks are comonotone.

\medskip
\subsection{Risk-Neutral Reinsurers}
\label{subsec:expectation}

Suppose that each reinsurer $j\in\cM$ wishes to minimize the expected risk according to their own subjective beliefs, represented by probability measures on the space $(S,\Sigma)$.
In other words, we assume that $\tau_j$ is a probability measure, for all $j\in\cM$.
In this case, each reinsurer's evaluation of the sum of contracts $\sum_{i=1}^nI_{ij}$ is the sum of the expectations of each contract, as shown by the following:
    \begin{align*}
        \rho_j^{Re}\left(\sum_{i=1}^n(1+\theta_j)I_{ij}(X_i)\right)
            &=\mathbb{E}^{\tau_j}\left(\sum_{i=1}^n(1+\theta_j)I_{ij}(X_i)\right)
        =(1+\theta_j)\sum_{i=1}^n \mathbb{E}^{\tau_j}\big(I_{ij}(X_i)\big)\,.
    \end{align*}

\noindent Our main result provides a characterization of SPNEs in this case.
Its proof is provided in Appendix \ref{appendix-proofs}.

\medskip

\begin{theorem}\label{ThSPNE}
    Suppose that $\tau_j$ is a probability measure for all $j\in\cM$.
    Let $\mathfrak{I}^*_i \in \aleph$ and $(\bm{\nu}^*_1,\ldots,\bm{\nu}^*_m) \in \beth$. Then the strategy $\stratstar$ is an SPNE.
\end{theorem}

\medskip
\subsection{Comonotone Risks}
\label{subsec:comonotone}

Alternatively, recall that since every $\rho_j^{Re}$ is a Choquet risk measure, they are comonotone additive.
We now show that if the initial risk random variables $X_1,\ldots,X_n$ are comonotone, then a similar characterization result to that of Theorem \ref{ThSPNE} is possible.
While it may not be realistic to assume that the initial risks of each insurer in the market are comonotone, recall that comonotonicity of losses represents a worst-case dependence structure for each reinsurer.
Consider a scenario where there is a probability measure $\Pr$ on the measurable space $(S,\Sigma)$, and that each $\rho_j^{Re}$ is law invariant with respect to $\Pr$.
That is, if two random variables $X$ and $Y$ have the same distribution under $\Pr$, we have 
    \[
        \rho_j^{Re}(X)=\rho_j^{Re}(Y)
    \]
for all $j\in\cM$.

Suppose that each reinsurer knows the distribution of each $X_i$ with respect to the measure $\Pr$, but does not necessarily know the dependence structure between these risks.
It would be prudent to evaluate risk based on the worst-case dependence in such a situation.
In other words, the reinsurer could evaluate the sum of contracts $\sum_{i=1}^nI_{ij}(X_i)$ by
    \[
        \rho_j^{Re}\left(\sum_{i=1}^n(1+\theta_j)I_{ij}(\tilde{X}_i)\right)\,,
    \]
where $(\tilde{X}_1,\ldots,\tilde{X}_n)$ is a comonotone vector of random variables such that $\tilde{X}_i$ and $X_i$ have the same distribution.
Such a comonotone vector is guaranteed to exist as long as the probability space $(S,\Sigma,\Pr)$ is atomless \cite[Lemma A.23]{follmer2011stochastic}.
Indeed, in  this case there exists a uniform random variable $U$ on the probability space $(S,\Sigma,\Pr)$, and taking $\tilde{X}_i:=F_{X_i,\Pr}^{-1}(U)$ for all $i\in\cN$ gives a comonotone vector of risks, where $F_{X_i,\Pr}^{-1}$ denotes a quantile function under $\Pr$ for $X_i$.
It then follows that
    \begin{align*}
        \rho_j^{Re}\left(\sum_{i=1}^n(1+\theta_j)I_{ij}(\tilde{X}_i)\right)
            &=(1+\theta_j)\sum_{i=1}^n\rho_j^{Re}(I_{ij}(\tilde{X}_i))
        =(1+\theta_j)\sum_{i=1}^n\rho_j^{Re}(I_{ij}(X_i))\,,
    \end{align*}
where the last equality follows from law invariance of $\rho_j^{Re}$.
Hence, in this scenario, it may be assumed without loss of generality that $X_1,\ldots,X_n$ are comonotone to begin with.
It then follows through similar arguments that a characterization of SPNE can be obtained, as given by the following.

\medskip

\begin{theorem}
\label{ThSPNEcom}
    Suppose that $X_1,\ldots,X_n$ are comonotone.
    Let $\mathfrak{I}^*_i \in \aleph$ and $(\bm{\nu}^*_1,\ldots,\bm{\nu}^*_m) \in \beth$. Then the strategy $\stratstar$ is an SPNE.
\end{theorem}

\begin{proof}
    Since $X_1,\ldots,X_n$ are comonotone and each indemnity function $\mathfrak{I}^*_{ij}(\nu^*_{i1},\cdots,\nu^*_{im})$ is non-decreasing, we have that for each $j\in\cM$, 
        \[
             \mathfrak{I}^*_{1j}(\nu^*_{11},\cdots,\nu^*_{1m})(X_1),\,
                \ldots,\,
                \mathfrak{I}^*_{nj}(\nu^*_{n1},\cdots,\nu^*_{nm})(X_n)
        \]
    are comonotone as well.
    Since $\rho_j^{Re}$ is comonotone additive, we have
        \begin{align*}
            &\quad\rho_j^{Re}\left(\sum_{i=1}^n (1+\theta_j)\mathfrak{I}^*_{ij}(\nu^*_{i1}, \cdots, \nu^*_{im})(X_i) \right) -  \sum_{i=1}^n \Pi^{\nu^*_{ij}}\left( \mathfrak{I}^*_{ij}(\nu^*_{i1}, \cdots, \nu^*_{im})(X_i)\right)\\            
            &=  \sum_{i=1}^n\rho_j^{Re}\bigg((1+\theta_j) \mathfrak{I}_{ij}^*(\nu^*_{i1}, \ldots, \nu^*_{im})(X_i)\bigg) -\sum_{i=1}^n \Pi^{\nu^*_{ij}}(\mathfrak{I}_{ij}^*(\nu^*_{i1}, \ldots, \nu^*_{im})(X_i)\big)\\
            &= \sum_{i=1}^n \bigg( \int (1+\theta_j)\mathfrak{I}_{ij}^*(\nu^*_{i1}, \ldots, \nu^*_{im})(X_i)\,d\tau_j - \int_0^M \nu^*_{ij}(X_i>z)\gamma^*_{ij}(z)\,dz\bigg)\,,
        \end{align*}

\noindent where each $\gamma^*_{ij}$ is the marginal indemnification of $I_{ij}^*$.
    The remainder of the proof follows similarly to that of Theorem \ref{ThSPNE}.
\end{proof}

\medskip
\subsection{Pareto Efficiency of SPNEs}
\label{SecPO}

Moreover, when SPNEs can be characterized using either Theorem \ref{ThSPNE} or \ref{ThSPNEcom}, these equilibria result in Pareto-optimal allocations to each agent.
Recall that an \emph{allocation} is a pair $\allocation\in\cI\times\mathbb{R}^{n\times m}$, representing both an admissible indemnity structure and the premium amount for each reinsurance contract.

\medskip

\begin{definition}
    \label{IRdef}
    An allocation $\allocationstar\in\cI\times\mathbb{R}^{n\times m}$ is individually rational (IR) if it incentivizes all agents to participate in the market.
    That is,
         \[
            \rho^{In}_i \left(X_i-\sum_{j=1}^mI_{ij}^*(X_i) +\sum_{j=1}^m\pi_{ij}^*\right) \leq \rho^{In}_i(X_i)
        \] 
    and 
        \[
            \rho^{Re}_j\left(\sum_{i=1}^n \big( (1+\theta_j) I_{ij}^*(X_i) - \pi_{ij}^*\big)\right) \leq 0\,,
        \]
    for all $i\in\cN$ and $j\in\cM$.
\end{definition}

\medskip

\begin{definition}
    \label{POdef}
    An allocation $\allocationstar\in\cI\times\mathbb{R}^{n\times m}$ is Pareto optimal (PO) if it is IR and there does not exist another allocation $\allocationtilde$ such that
    \begin{align*}
        \rho^{In}_i\left(X_i-\sum_{j=1}^m\tilde{I}_{ij}(X_i) +\sum_{j=1}^m\tilde{\pi}_{ij}\right)
            &\leq \rho^{In}_i \left(X_i-\sum_{j=1}^mI_{ij}^*(X_i) +\sum_{j=1}^m\pi_{ij}^*\right)\\
        \rho^{Re}_j\left(\sum_{i=1}^n\big((1+\theta_j) \tilde{I}_{ij}(X_i) - \tilde{\pi}_{ij}\big)\right)
            &\leq \rho^{Re}_j\left(\sum_{i=1}^n \big((1+\theta_j)I_{ij}^*(X_i)-\pi_{ij}^*\big)\right)
    \end{align*} 
    for all $i\in\cN$ and $j\in\cM$, with at least one strict inequality.
\end{definition}

The following provides an explicit characterization of PO allocations in this market. Its proof can be found in Appendix \ref{appendix-proofs}.

\medskip

\begin{proposition}
\label{prop:optimal_indemnities}
    Suppose that either $\tau_j$ is a probability measure for all $j\in\cM$, or that $X_1,\ldots,X_n$ are comonotone.
    Then an allocation $\allocationstar\in\cI\times\mathbb{R}^{n\times m}$ is Pareto optimal if and only if the following hold:
        \begin{enumerate}[label=(\roman*)]
            \item \label{prop:optimal_indemnities_1}
                The marginal indemnifications $\gamma_{ij}^*$ satisfy
                    \[
                        \sum_{j \in \cL_i(z)}\gamma_{ij}^*(z) =1\quad\mbox{and}\quad\sum_{j \in\cM\setminus\cL_i(z)}\gamma_{ij}^*(z) =0\,,
                    \]
                    for each $i\in\cN$, where
                    \[
                        \cL_i(z):= \left\{j \in \cM\cup \{0\}\,\,\middle|\,\, \tau_{i,j}(X_i>z) = \min_{k \in \cM \cup \{0\}} \tau_{i,k}(X_i>z)\right\}\,.
                    \]
            \item \label{prop:optimal_indemnities_2}
                The premia $\{\pi^*_{ij}\}_{i\in\cN,j\in\cM}$ are chosen such that $\allocationstar$ is individually rational.
        \end{enumerate}
\end{proposition}

\medskip

\begin{theorem}
\label{thm:induce_IRPO}
    Suppose that either $\tau_j$ is a probability measure for all $j\in\cM$, or that $X_1,\ldots,X_n$ are comonotone.
    Let $\mathfrak{I}^*_i\in\aleph$ and $(\bm{\nu}^*_{1},\ldots,\bm{\nu}^*_{m})\in\beth$.
    Then the strategy $\stratstar$ induces a PO allocation.
    That is, the allocation
        \[
            \left(\left\{\mathfrak{I}^*_{ij}(\nu^*_{i1},\ldots,\nu^*_{im})\right\}_{i\in \cN, j \in \cM},\,\,
                \left\{\Pi^{\nu^*_{ij}}\left(\mathfrak{I}^*_{ij}(\nu^*_{i1},\ldots,\nu^*_{im})(X_i)\right)\right\}_{i\in \cN, j \in \cM}\right)
        \]
    is Pareto optimal.
\end{theorem}

\begin{proof}
    By Proposition \ref{prop:optimal_indemnities}, verifying Pareto optimality in this market involves checking the two conditions \ref{prop:optimal_indemnities}\ref{prop:optimal_indemnities_1} and \ref{prop:optimal_indemnities}\ref{prop:optimal_indemnities_2}.
    Note first that by construction of the strategy $\stratstar$, it is easily verified that for each $i\in\cN$, we have that
        \[    \big(\nu_{i1}^*,\ldots,\nu_{im}^*,\mathfrak{I}^*_i\big)
        \]
    is a strategy for a market with one insurer and $m$ reinsurers that satisfies the conditions of \cite[Theorem 3.11]{zhu2023equilibria}.
    It then follows from \cite[Theorem 4.4]{zhu2023equilibria} that for each $i\in\cN$, we have
        \begin{align*}
            \rho^{In}_i\left(X_i-\sum_{j=1}^m\mathfrak{I}^*_{ij}(\nu^*_{i1},\ldots,\nu^*_{im})(X_i)+\sum_{j=1}^m\Pi^{\nu^*_{ij}}\left(\mathfrak{I}^*_{ij}(\nu^*_{i1},\ldots, \nu^*_{im})(X_i)\right)\right)&\le \rho^{In}_i(X_i)\,,\\
            \rho^{Re}_j\left( (1+\theta_j)\mathfrak{I}^*_{ij}(\nu^*_{i1},\ldots,\nu^*_{im})(X_i)-\Pi^{\nu^*_{ij}}\left(\mathfrak{I}^*_{ij}(\nu^*_{i1},\ldots, \nu^*_{im})(X_i)\right)\right)&\le0\quad\forall j\in\cM\,.
        \end{align*}
    
\noindent By summing the second inequality over all $i\in\cN$, we have that
        \[
      \rho^{Re}_j\left(\sum_{i=1}^n \left((1+\theta_j)\mathfrak{I}^*_{ij}(\nu^*_{i1},\ldots,\nu^*_{im})(X_i)-\Pi^{\nu^*_{ij}}\left(\mathfrak{I}^*_{ij}(\nu^*_{i1},\ldots, \nu^*_{im})(X_i)\right)\right)\right)\le0\quad\forall j\in\cM\,,
        \]
    implying individual rationality, so condition \ref{prop:optimal_indemnities}\ref{prop:optimal_indemnities_2} holds.
    
    It remains to verify condition \ref{prop:optimal_indemnities}\ref{prop:optimal_indemnities_1}.
    To this end, notice that this condition is verified by iterating over each $i\in\cN$.
    It then follows that a direct application of \cite[Theorem 4.12]{zhu2023equilibria} to each $i\in\cN$ yields the desired result.
\end{proof}

\bigskip
\section{Stackelberg Equilibrium vs.\ SPNEs -- A Numerical Example}
\label{sec:example}

In the following, we show that the introduction of additional reinsurers into the market has a positive effect on the welfare of the insurers.
When there is only one reinsurer in the market (i.e., $m=1$), we recover the setting of \cite{ghossoub2024stackelberg}, who show that none of the $n$ insurers are able to experience a welfare gain, because of monopoly.
In the following example, we will show that increasing the number of reinsurers can alleviate this issue by capturing value for the insurers.

\medskip
\subsection{The Stackelberg Model}

The market with one reinsurer with a first-mover advantage is the Stackelberg model, for which the solution concept is given by the following.
Note that since $m=1$ in this model, we drop the letter $j$ from the notation for simplicity.

\medskip

\begin{definition}
    A pair $\left(\{I^*_i\}_{i\in\cN},\{\nu^*_i\}_{i\in\cN}\right) \in \mathcal{I}\times\cC^n$ is said to be a \emph{Stackelberg equilibrium} (SE) if
        \begin{itemize}
            \item  For each $i\in\cN$, we have
                    $I^*_{i}
                    \in \underset{I_{i}\in\cI_i}{\arg\min}\,
                    \rho_i^{In}\big( X_i - I_{i}(X_i) + \Pi^{\nu_i^*}(I_{i}(X_i))\big)$\,, and,
            
            \item $ \rho^{Re}\big(\sum_{i=1}^n((1+\theta_1)I^*_{i}(X_i)-\Pi^{\nu_i^*}(I^*_{i}(X_i))\big) 
                    \leq\rho^{Re}\big(\sum_{i=1}^n((1+\theta_1)\tilde{I}_{i}(X_i)-\Pi^{\tilde{\nu}_i}(\tilde{I}_{i}(X_i))\big)$
                for all $\{\tilde{\nu}_{i}\}_{i=1}^n\in\cC^n$ and $\{\tilde{I}_{i}\}_{i\in\cN} \in\cI$ such that      
                    \[
                        \tilde{I}_{i}
                        \in\underset{I_i\in\cI_i}{\arg\min}\,\,\rho_i^{In}
                        \left(X_i-I_{i}(X_i)+\Pi^{\tilde{\nu}_i}(I_{i}(X_i)\right)
                        \quad\forall i\in\cN\,.
                    \]
        \end{itemize}
\end{definition}

It follows from this definition that Stackelberg equilibria are a special case of SPNEs.
Hence by \cite[Proposition 5.3]{ghossoub2024stackelberg}, if $(\{I^*_{i}\}_{i \in \cN},\{\nu^*_i\}_{i\in \cN} )$ is an SE, then for each $i\in\cN$, we have
    \[
        \rho_i^{In}\left(X_i-I^*_{i}(X_i)+\Pi^{\nu_i^*}(I^*_{i}(X_i)\right)=\rho_i^{In}(X_i)\,,
    \]

\noindent implying that each insurer $i$ is indifferent between reinsurance and no reinsurance.
This is further illustrated by the following numerical example. Consider a reinsurance market with $n=3$ insurers and $m=1$ reinsurer. 
We assume that for each $i\in\cN=\{1,2,3\}$, the risk $X_i$ is distributed according to a censored exponential distribution with respect to a given probability measure $\Pr$.
Each $X_i$ has parameter $\lambda=3$ and is censored at a value of $5$.
Hence, their survival functions are given by
    \[
        \mathbb{P}(X_i>z) = \begin{cases}
            e^{-\lambda z} \,\, & 0 \leq z \leq 5 \\
            0 \,\, & z>5
        \end{cases}\,.
    \]

\noindent The reinsurer is assumed to be risk neutral, with subjective beliefs represented by a probability measure $\tau$ and a loading factor of $\theta=0.15$.
We assume that the reinsurer judges each $X_i$ to be a censored exponential distribution, but with a different parameter value of $\beta=2.5$.
The survival function of the reinsurer's preference is therefore
    \[
        \tau(X_i>z)=\begin{cases}
            e^{-\beta z}\,\,& 0 \leq z \leq 5 \\
            0 \,\, & z>5
            \end{cases}\,,
    \]
for each $i\in\cN$.

We will assume that the preferences for each insurer are given by the Expected Shortfall.
We recall the following standard definitions:

\medskip

\begin{definition}
    The \emph{Value-at-Risk (VaR)} at level $\gamma\in(0,1)$ of a random variable $X\in\mathcal{X}$ under the probability measure $\Pr$ is
        \[
        \mathrm{VaR}_\gamma^\Pr(X)
                :=\inf\left\{t\in\R\,\,\middle|\,\,\Pr(X>t)\le\gamma\right\}\,.
        \]
\end{definition}

\medskip

\begin{definition}
    The \emph{Expected Shortfall (ES)} at level $\gamma\in(0,1)$ of a random variable $X\in\mathcal{X}$ under the probability measure $\Pr$ is
        \[
            \mathrm{ES}_\gamma^\Pr(Z)
                :=\frac{1}{\gamma}\int_0^\gamma\mathrm{VaR}_u^\Pr(X)\,du\,.
        \]
\end{definition}

We assume that risk measure of each insurer $i$ is an ES with respect to the parameter $\gamma_i$, with $\gamma_1=10\%$, $\gamma_2=5\%$, and $\gamma_3=1\%$.
It is well-known that the ES admits the following representation as a Choquet risk measure.
For all $Z\in\cX$,
    \[
        \rho_i^{In}(Z)=\mathrm{ES}_{\gamma_i}^\Pr(Z)=\int Z\,dT_i\circ\Pr=\int Z\,d\alpha_i\,,
    \]
where
    \[
        T_i(t):=\min\left\{t/\gamma_i, 1\right\}\quad\forall i\in\cN\,.
    \]

As a special case of Theorem \ref{ThSPNE}, we may solve for the indemnity structure $I_1^*,I_2^*,I_3^*$, as shown in Figure \ref{Fig5}.
The risk to each agent both before and after reinsurance is given in Table \ref{table:stackelberg}.
Notably, the initial risk and post-transfer risk are identical for each of the insurers, indicating a lack of welfare gain from reinsurance.

\begin{figure}[hbtp!]
    \centering
\includegraphics[width=0.5\linewidth]{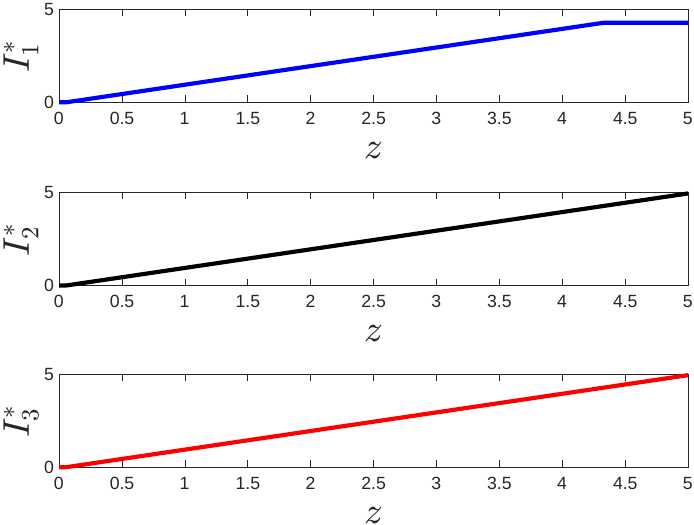}
    \caption{Optimal Indemnity Structure.}
    \label{Fig5}
\end{figure}
\begin{table}[hbtp!]
    \centering
    \begin{tabular}{|c|c|c|c|}
    \hline
         & \textbf{Initial Risk} & \textbf{Premium Paid} & \textbf{Post-Transfer Risk}\\
        \hline
        Insurer 1 & 1.100861  & 1.044949 & 1.100861 \\
        Insurer 2 & 1.331909  & 1.276004 & 1.331909 \\
        Insurer 3 & 1.868380  & 1.812475 & 1.868380 \\
        Reinsurer & 0 & -- & -2.933441 \\ 
        \hline
    \end{tabular}
    \caption{Results for the Stackelberg Model.}
    \label{table:stackelberg}
\end{table}

\medskip
\subsection{SPNE with Multiple Reinsurers}

We now examine the effect of the introduction of an additional reinsurer to the market described above.
That is, we take $m=2$, and we suppose that the second reinsurer has a loading factor of $\theta_2=0.10$, and believes that each risk has a parameter value of $\beta_2=2$.
The survival functions for each $X_i$ according to the reinsurers are therefore given by
    \[
        \tau_j(X_i>z)=\begin{cases}
            e^{-\beta_jz}\,\,& 0 \leq z \leq 5 \\
            0 \,\, & z>5
            \end{cases}\,,
    \]
with $\beta_1=2.5$ and $\beta_2=2$.

By Theorem \ref{ThSPNE}, the explicit characterization of SPNE depends on the relationship between the functions $\alpha_i$ and $\tau_{j}$. Figure \ref{Fig1} plots these functions, and Figure \ref{Fig2} plots the second-lowest true preferences $\bar{\tau}_i$.
The SPNE indemnity structure is given in Figure \ref{Fig3}.

\begin{figure}[hbtp!]
	\centering
    \subfloat[Survival Probabilities $\alpha_i$ and $\tau_{j}$.]{\includegraphics[width=0.5\textwidth]{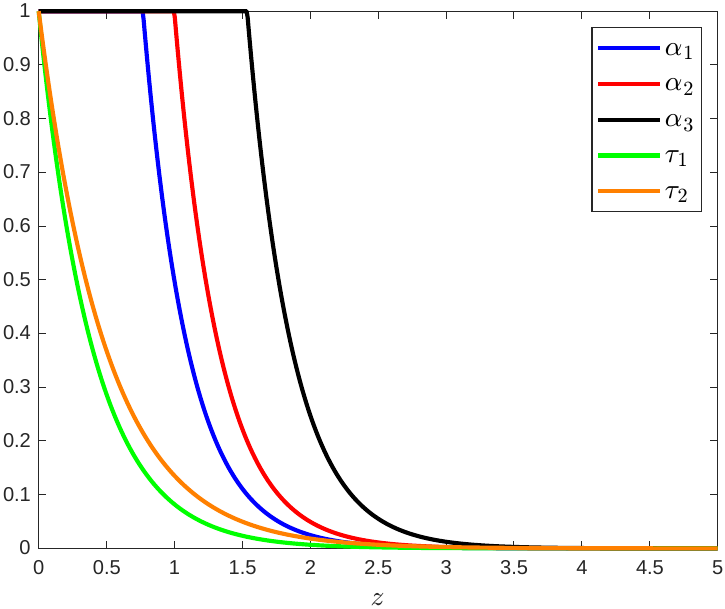}\label{Fig1}}
    \subfloat[Second-Lowest Preference $\bar{\tau}_i$.]{\includegraphics[width=0.5\textwidth]{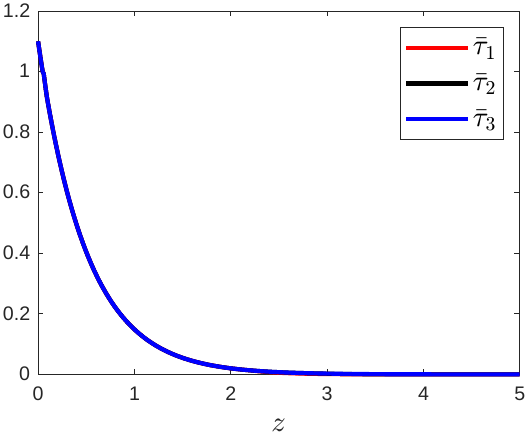}\label{Fig2}}
    \caption{\vspace{0.8cm}}
\end{figure}
\vspace{-1cm}
\begin{figure}[H]
    \centering
\includegraphics[width=0.5\linewidth]{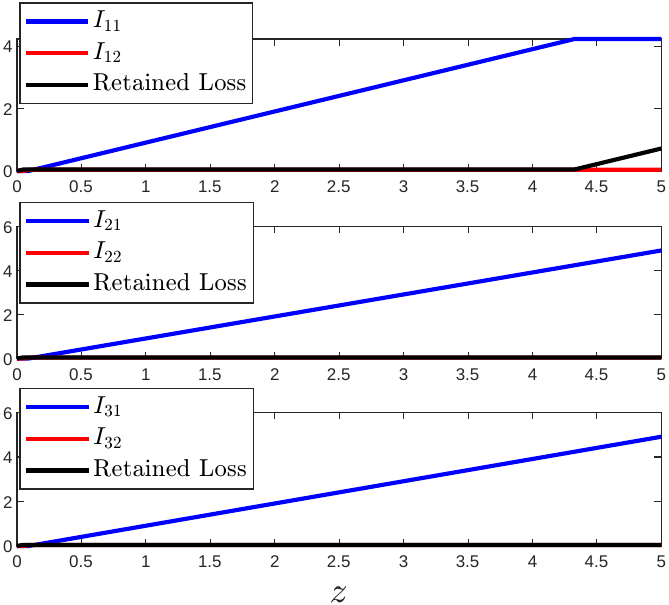}
    \caption{Optimal Indemnity Structure For Each Insurer}
    \label{Fig3}
\end{figure}

A numerical summary of the allocations resulting from the SPNE is given in Table \ref{Tab1}.
Note that in contrast to the Stackelberg case, every agent experiences a welfare gain in the case of $m=2$ reinsurers.
We interpret this as a positive consequence of the supply-side competition when a second reinsurer is introduced.
Indeed, the introduction of an additional reinsurer prevents any single reinsurer from charging extremely high premia, since doing so would invite the other supplier to undercut their price.

\begin{table}[H]
    \centering
    \begin{tabular}{|c|c|c|c|}
        \hline
        & & \textbf{One Reinsurer} & \textbf{Two Reinsurers} \\
        \hline
                  & Initial Risk       & 1.100861 & 1.100861 \\
        Insurer 1 & Post-Transfer Risk & 1.100861 & 0.545767 \\
                  & Welfare Gain       & 0        & 0.555094 \\
        \hline
                  & Initial Risk       & 1.331909 & 1.331909 \\
        Insurer 2 & Post-Transfer Risk & 1.331909 & 0.547430 \\
                  & Welfare Gain       & 0        & 0.784479 \\
        \hline
                  & Initial Risk       & 1.868380 & 1.868380 \\
        Insurer 3 & Post-Transfer Risk & 1.868380 & 0.547954 \\
                  & Welfare Gain       & 0        & 1.320426 \\
        \hline 
        \multirow{2}{*}{Reinsurer 1} & Post-Transfer Risk & -2.933441 & -0.273442 \\
                                     & Welfare Gain       &  2.933441 &  0.273442 \\
        \hline
        \multirow{2}{*}{Reinsurer 2} & Post-Transfer Risk & --        & -0.000993 \\
                                     & Welfare Gain       & --        &  0.000993 \\

        \hline
    \end{tabular}
    \caption{Results with Two Reinsurers}
    \label{Tab1}
\end{table}

\bigskip
\section{Conclusion}\label{sec:conclusion}

In this paper, we provide a unifying framework for the study of equilibria in large sequential-move reinsurance markets populated by several insurers on the demand side and several reinsurers on the supply side, where the reinsurers have the first-mover advantage. Each insurer has preferences represented by a general Choquet risk measure and can purchase coverage from any or all reinsurers. Each reinsurer has preferences represented by a general Choquet risk measure and can provide coverage to any or all insurers. Pricing in this market is done via a nonlinear pricing rule given by a Choquet integral, and we allow for belief heterogeneity.

Within this framework, we propose that the Subgame Perfect Nash Equilibrium (SPNE) is the most suitable equilibrium concept. We show that SPNEs can be determined in terms of backward induction, and we provide a characterization thereof in the case of risk-neutral reinsurers and in the case of comonotone losses. In both cases, we show that equilibrium contracts are individually rational and Pareto optimal. Furthermore, in a numerical example, we show that equilibrium contracts in our market induce a strict welfare gain for the insurers, as opposed to the monopoly market with multiple agents on the demand side examined by  \cite{ghossoub2024stackelberg}.


\newpage
\hypertarget{LinkToAppendix}{\ }
\appendix

\vspace{-0.8cm}


\section{Proofs of Main Results} \label{appendix-proofs}

\medskip
\subsection{Proof of Theorem \ref{ThSPNE}}

For $i \in \mathcal{N}$ and $j \in \cM$, let $\mathfrak{I}^*_i \in \aleph$ and $(\bm{\nu}_{1},\ldots,\bm{\nu}_m)\in\beth$.
To verify that this is strategy indeed an SPNE, we apply Lemma \ref{SPNEminprop}.
By the definition of $\aleph$, it follows immediately that $\mathfrak{I}^*_i(\nu^*_{i1},\ldots,\nu^*_{im})$ is a solution to \eqref{eq:minproblemSPNE}.
It remains to show that there does not exist a $\bm{\tilde{\nu}}\in\cC^n$ such that \eqref{eq:no_deviation} holds.
To this end, we will first show that the total post-transfer risk to each reinsurer $j$ is exactly
    \[
        \sum_{i=1}^n\int_0^M\min\{\tau_{i,j}(X_i>z)-\bar{\tau}_i(X_i>z),\,0\}\,dz\,.
    \]
We then argue that it is not possible for each reinsurer to improve their post-transfer risk beyond this value by changing their pricing strategy.

To this end, recall that we may write each indemnity function $\mathfrak{I}^*_{ij}(\nu^*_{i1},\ldots,\nu^*_{im})$ in terms of its marginal indemnification $\gamma_{ij}^*$, as in \eqref{eq:gamma_defn}.
    \begin{align*}
    \rho_j^{Re}&\left(\sum_{i=1}^n (1+\theta_j)\mathfrak{I}^*_{ij}(\nu^*_{i1}, \cdots, \nu^*_{im})(X_i) \right) -  \sum_{i=1}^n \Pi^{\nu^*_{ij}}\left( \mathfrak{I}^*_{ij}(\nu^*_{i1}, \cdots, \nu^*_{im})(X_i)\right)\\
        &=  \sum_{i=1}^n \E^{\tau_j}\bigg((1+\theta_j) \mathfrak{I}_{ij}^*(\nu^*_{i1}, \ldots, \nu^*_{im})(X_i)\bigg) -\sum_{i=1}^n \Pi^{\nu^*_{ij}}(\mathfrak{I}_{ij}^*(\nu^*_{i1}, \ldots, \nu^*_{im})(X_i)\big)\\
        &= \sum_{i=1}^n \bigg(\E^{\tau_j}\big((1+\theta_j) \mathfrak{I}_{ij}^*(\nu^*_{i1}, \ldots, \nu^*_{im})(X_i)\big) - \Pi^{\nu^*_{ij}}(\mathfrak{I}_{ij}^*(\nu^*_{i1}, \ldots, \nu^*_{im})(X_i)\big)\bigg)\\
        &= \sum_{i=1}^n \bigg( \int (1+\theta_j)\mathfrak{I}_{ij}^*(\nu^*_{i1}, \ldots, \nu^*_{im})(X_i)\,d\tau_j - \int_0^M \nu^*_{ij}(X_i>z)\gamma^*_{ij}(z)\,dz\bigg)\\
        &= \sum_{i=1}^n \bigg(\int_0^M(1+\theta_j)\tau_j(X_i>z) \gamma^*_{ij}(z)\,dz - \int_0^M \nu^*_{ij}(X_i>z)\gamma^*_{ij}(z)\,dz\bigg)\\
        &= \sum_{i=1}^n \bigg(
        \int_0^M\tau_{i,j}(X_i>z) \gamma^*_{ij}(z)\,dz - \int_0^M \nu^*_{ij}(X_i>z)\gamma^*_{ij}(z)\,dz\bigg)\\
        &=  \sum_{i=1}^n \int_0^M \big( \tau_{i,j}(X_i>z) - \nu^*_{ij}(X_i>z)\big) \gamma^*_{ij}(z)\,dz\,.
    \end{align*}

By repeatedly applying \cite[Proposition 3.9]{zhu2023equilibria} for each $i\in\cN$, it follows that
    \begin{equation}
        \label{eq:target}
        \sum_{i=1}^n \int_0^M \big( \tau_{i,j}(X_i>z) - \nu^*_{ij}(X_i>z)\big) \gamma^*_{ij}(z)\,dz
        =\sum_{i=1}^n\int_0^M\min\{\tau_{i,j}(X_i>z)-\bar{\tau}_i(X_i>z),\,0\}\,dz\,.
    \end{equation}
In other words, the improvement in the reinsurer's risk measure is precisely the sum over all contracts of the positive difference between their true preference $\tau_{i,j}$ and the second-lowest preference $\bar{\tau}_i$.
That is, the reinsurer profits exactly when their true preference $\tau_{i,j}$ is lower than $\bar{\tau}_i$.

Now for a fixed reinsurer $j$, consider an alternative strategy $\bm{\tilde{\nu}}\in\cC^n$.
If the reinsurer were to switch to this alternative strategy, the indemnity structure selected by each insurer $i$ would be $\mathfrak{I}^*_{ij}(\tilde{\nu}_i,\nu^*_{i, -j})$, which can be expressed in terms of the marginal indemnification as
    \[
        \mathfrak{I}^*_{ij}(\tilde{\nu}_i, \nu^*_{i, -j})(x) = \int_0^x \tilde{\gamma}_{ij}(z)\,dz\,.
    \]

Using a similar simplification as the above, we can show that
    \begin{align*}
    \rho_j^{Re}&\left(\sum_{i=1}^n (1+\theta_j) \mathfrak{I}^*_{ij}(\tilde{\nu}_i, \nu^*_{i,-j})(X_i)\right)
            -\sum_{i=1}^n \Pi^{\tilde{\nu}_i}\left(\mathfrak{I}^*_{ij}(\tilde{\nu}_i, \nu^*_{i,-j})(X_i) \right)\\
        &=\sum_{i=1}^n \int_0^M \big( \tau_{i,j}(X_i>z) -\tilde{\nu}_i(X_i>z)\big)\tilde{\gamma}_{ij}(z)\,dz\,.
    \end{align*}
Hence, by \eqref{eq:target}, it suffices to show that
    \begin{align*}
        \sum_{i=1}^n\int_0^M \big( \tau_{i,j}(X_i>z) - \tilde{\nu}_i(X_i>z)\big)\tilde{\gamma}_{ij}(z)\,dz
            &\ge\sum_{i=1}^n\int_0^M \big( \tau_{i,j}(X_i>z) - \nu^*_{ij}(X_i>z)\big) \gamma^*_{ij}(z)\,dz\\
        &=\sum_{i=1}^n\int_0^M\min\{\tau_{i,j}(X_i>z)-\bar{\tau}_i(X_i>z),\,0\}\,dz.
    \end{align*}

Fix $i\in\cN$, and consider the following cases for different tail risks $X_i>z$.
    \begin{enumerate}
        \item $\tau_{i,j}(X_i>z)-\bar{\tau}_i(X_i>z)<0$:
            In this case, the reinsurer is fully profiting from indemnifying this tail risk when using the original strategy $\nu_{ij}^*$.
            If the reinsurer were to switch their strategy to $\tilde{\nu}_i$, we have the following:
            \begin{itemize}
                \item $\tilde{\nu}_i(X_i>z)>\bar{\tau}_i(X_i>z)$: In this case the reinsurer has increased the price for this tail risk. 
                However, by point (2) of Definition \ref{defn:beth}, there is another reinsurer charging the price $\bar{\tau}_i(X_i>z)$, so the new price $\tilde{\nu}_i(X_i>z)$ is no longer the lowest price in the market.
                This implies that $\tilde{\gamma}_{ij}=0$, and so
                    \[
                        \big( \tau_{i,j}(X_i>z) - \tilde{\nu}_i(X_i>z)\big)\tilde{\gamma}_{ij}(z)=0\ge\min\{\tau_{i,j}(X_i>z)-\bar{\tau}_i(X_i>z),\,0\}\,.
                    \]
                \item $\tilde{\nu}_i(X_i>z)<\bar{\tau}_i(X_i>z)$: In this case the reinsurer has decreased the price for this tail risk.
                By point (1) of Definition \ref{defn:beth}, this price $\tilde{\nu}_i(X_i>z)$ is now the unique lowest price for this tail risk, which implies $\tilde{\gamma}_{ij}(z)=1$. Hence,
                    \begin{align*}
                        \big( \tau_{i,j}(X_i>z)-\tilde{\nu}_i(X_i>z)\big)\tilde{\gamma}_{ij}(z)&=\tau_{i,j}(X_i>z)-\tilde{\nu}_i(X_i>z)\\
                        &\ge\tau_{i,j}(X_i>z)-\bar{\tau}_i(X_i>z)\\
                        &=\min\{\tau_{i,j}(X_i>z)-\bar{\tau}_i(X_i>z),\,0\}\,.
                    \end{align*}
                \item $\tilde{\nu}_i(X_i>z)=\bar{\tau}_i(X_i>z)$: In this case the reinsurer has not changed the price for this tail risk. Here, we have
                    \begin{align*}
                        \big( \tau_{i,j}(X_i>z)-\tilde{\nu}_i(X_i>z)\big)\tilde{\gamma}_{ij}(z)
                            &=\big( \tau_{i,j}(X_i>z)-\bar{\tau}_i(X_i>z)\big)\tilde{\gamma}_{ij}(z)\\
                        &\ge\tau_{i,j}(X_i>z)-\bar{\tau}_i(X_i>z)\\
                        &=\min\{\tau_{i,j}(X_i>z)-\bar{\tau}_i(X_i>z),\,0\}\,.
                    \end{align*}
            \end{itemize}

        \item $\tau_{i,j}(X_i>z)-\bar{\tau}_i(X_i>z)\ge0$: In this case, the reinsurer is unable to profit from indemnifying this tail risk using the original strategy. In what follows, we show that it is not possible to profit on these tail risks regardless of the choice of pricing capacity.
            \begin{itemize}
                \item $\tilde{\nu}_i(X_i>z)>\bar{\tau}_i(X_i>z)$: Similar to case (1), increasing the price for this tail risk causes $\tilde{\gamma}_{ij}=0$, so
                    $
                    \big(\tau_{i,j}(X_i>z) - \tilde{\nu}_i(X_i>z)\big)\tilde{\gamma}_{ij}(z)=0=\min\{\tau_{i,j}(X_i>z)-\bar{\tau}_i(X_i>z),\,0\}\,.
                    $
                \item $\tilde{\nu}_i(X_i>z)\le\bar{\tau}_i(X_i>z)$: In this case the reinsurer has decreased the price for this tail risk, but since $\tilde{\nu}_i(X_i>z)\le\bar{\tau}_i(X_i>z)\le\tau_{i,j}(X_i>z)$, this price is now so low that a loss is incurred to reinsurer $j$. That is,
                    \[
                    \big(\tau_{i,j}(X_i>z) - \tilde{\nu}_i(X_i>z)\big)\tilde{\gamma}_{ij}(z)\ge 0=\min\{\tau_{i,j}(X_i>z)-\bar{\tau}_i(X_i>z),\,0\}\,.
                    \]
            \end{itemize}
    \end{enumerate}

    \noindent Hence, we have shown that for all $i\in\cN$, we have
        \[
            \big(\tau_{i,j}(X_i>z) - \tilde{\nu}_i(X_i>z)\big)\tilde{\gamma}_{ij}(z)\ge\min\{\tau_{i,j}(X_i>z)-\bar{\tau}_i(X_i>z),\,0\}\,.
        \]
    It then follows that
        \[
        \sum_{i=1}^n\int_0^M \big( \tau_{i,j}(X_i>z) - \tilde{\nu}_i(X_i>z)\big)\tilde{\gamma}_{ij}(z)\,dz
            \ge\sum_{i=1}^n\int_0^M\min\{\tau_{i,j}(X_i>z)-\bar{\tau}_i(X_i>z),\,0\}\,dz\\
        \]
    as desired, which completes the proof.\qed

\bigskip
\subsection{Proof of Proposition \ref{prop:optimal_indemnities}}

We prove this characterization result for Pareto-optimal allocations using two steps.
Firstly, since preferences are translation invariant, Pareto optimality of an allocation is equivalent to minimizing the aggregate post-transfer risk, as shown by Lemma \ref{lem:po_minimize_risk} below.
It then follows that those allocations that minimize aggregate post-transfer risk are precisely those of the form given in the statement of Proposition \ref{prop:optimal_indemnities}.
To this end, we first show the following:

\medskip

\begin{lemma}
    \label{lem:po_minimize_risk}
    An allocation $\allocationstar\in\cI\times\mathbb{R}^{n\times m}$ is PO if and only if the following hold:
        \begin{enumerate}[label=(\roman*)]
            \item \label{lem:po_minimize_risk_1}
                The indemnity structure $\{I^*_{ij}\}_{i\in\cN,j\in\cM}$ is a solution to
        \begin{equation}\label{infPO}
                        \inf_{\{I_{ij}\}_{i\in\cN,j\in\cM}\in\cI}
                            \left\{\sum_{j=1}^m \rho^{Re}_j\left(\sum_{i=1}^n(1+\theta_j)I_{ij}(X_i)\right)
                    +\sum_{i=1}^n \rho^{In}_i \left(X_i - \sum_{j=1}^m I_{ij}(X_i)\right)\right\}\,.
                    \end{equation}
            \item \label{lem:po_minimize_risk_2}
                The premia $\{\pi^*_{ij}\}_{i\in\cN,j\in\cM}$ are chosen such that $\allocationstar$ is IR.
        \end{enumerate}
\end{lemma}

\begin{proof}
    Let $\allocationstar$ be an allocation such that \ref{lem:po_minimize_risk_1} and \ref{lem:po_minimize_risk_2} hold, and assume for the sake of contradiction that this allocation is not PO.
    Then there exists another allocation $\allocationtilde$ such that
        \begin{align*}
            \rho^{In}_i\left(X_i-\sum_{j=1}^m\tilde{I}_{ij}(X_i)+\sum_{j=1}^m\tilde{\pi}_{ij}\right)
                &\le\rho^{In}_i \left(X_i - \sum_{j=1}^m I^*_{ij}(X_i) +\sum_{j=1}^m\pi^*_{ij}\right)\quad\forall i\in\cN\,,\\
            \rho^{Re}_j\left(\sum_{i=1}^n\big((1+\theta_j)\tilde{I}_{ij}(X_i)-\tilde{\pi}_{ij}\big)\right)
                &\le\rho^{Re}_j\left(\sum_{i=1}^n\big((1+\theta_j) I^*_{ij}(X_i)-\pi^*_{ij}\big)\right)\quad\forall j\in\cM\,,
        \end{align*}
    with at least one strict inequality for some $i\in\cN$ or $j\in\cM$.
    Summing these inequalities gives
        \begin{align*}
            \sum_{j=1}^m\rho^{Re}_j&\left(\sum_{i=1}^n\big((1+\theta_j)\tilde{I}_{ij}(X_i)-\tilde{\pi}_{ij}\big)\right)
                +\sum_{i=1}^n\rho^{In}_i\left(X_i-\sum_{j=1}^m\tilde{I}_{ij}(X_i)+\sum_{j=1}^m\tilde{\pi}_{ij}\right)\\
            &<\sum_{j=1}^m \rho^{Re}_j\left(\sum_{i=1}^n \big((1+\theta_j)I^*_{ij}(X_i) - \pi^*_{ij}\big)\right) + \sum_{i=1}^n \rho^{In}_i \left(X_i - \sum_{j=1}^mI^*_{ij}(X_i) +\sum_{j=1}^m\pi^*_{ij}\right)\,.
        \end{align*}
    
    \noindent Since each $\rho_i^{In}$ and $\rho_j^{Re}$ is translation invariant, the premia $\tilde{\pi}_{ij}$ and $\pi_{ij}^*$ all cancel out, which yields
        \begin{align*}
            \sum_{j=1}^m\rho^{Re}_j&\left(\sum_{i=1}^n(1+\theta_j)\tilde{I}_{ij}(X_i)\right)
                +\sum_{i=1}^n\rho^{In}_i\left(X_i-\sum_{j=1}^m\tilde{I}_{ij}(X_i)\right)\\
            &<\sum_{j=1}^m\rho^{Re}_j\left(\sum_{i=1}^n(1+\theta_j)I^*_{ij}(X_i)\right)+\sum_{i=1}^n\rho^{In}_i\left(X_i-\sum_{j=1}^mI^*_{ij}(X_i)\right)\,,
        \end{align*}
    contradicting \ref{lem:po_minimize_risk_1}. Hence, $\allocationstar$ must be PO.

    Conversely, suppose that $\allocationstar$ is PO.
    Since PO allocations are IR by definition, \ref{lem:po_minimize_risk_2} holds.
    Assume for the sake of contradiction that \ref{lem:po_minimize_risk_1} does not hold, so there exists an indemnity structure $\{\tilde{I}_{ij}\}_{i\in\cN,j\in\cM}\in\cI$ such that
        \begin{align*}
            \sum_{j=1}^m\rho^{Re}_j&\left(\sum_{i=1}^n(1+\theta_j)\tilde{I}_{ij}(X_i)\right)
                +\sum_{i=1}^n\rho^{In}_i\left(X_i-\sum_{j=1}^m\tilde{I}_{ij}(X_i)\right)\\
            &<\sum_{j=1}^m\rho^{Re}_j\left(\sum_{i=1}^n(1+\theta_j)I^*_{ij}(X_i)\right)+\sum_{i=1}^n\rho^{In}_i\left(X_i-\sum_{j=1}^mI^*_{ij}(X_i)\right)\,.
        \end{align*}
    
\noindent Let $\varepsilon>0$ denote the positive difference in the above inequality.
    That is,
        \begin{align*}
            \varepsilon&:=\sum_{j=1}^m\rho^{Re}_j\left(\sum_{i=1}^n(1+\theta_j)I^*_{ij}(X_i)\right)
                +\sum_{i=1}^n\rho^{In}_i\left(X_i-\sum_{j=1}^mI^*_{ij}(X_i)\right)\\
            &\quad\quad-\sum_{j=1}^m\rho^{Re}_j\left(\sum_{i=1}^n(1+\theta_j)\tilde{I}_{ij}(X_i)\right)
                -\sum_{i=1}^n\rho^{In}_i\left(X_i-\sum_{j=1}^m\tilde{I}_{ij}(X_i)\right)\,.
        \end{align*}
    Define
        \begin{align*}
            a_i&:=\sum_{j=1}^m\pi_{ij}^*
                +\rho^{In}_i\left(X_i-\sum_{j=1}^mI_{ij}^*(X_i)\right)
                -\rho^{In}_i\left(X_i-\sum_{j=1}^m\tilde{I}_{ij}(X_i)\right)
                \quad\forall i\in\cN\,,\\
            b_j&:=\sum_{i=1}^n\pi_{ij}^*
                +\rho^{Re}_j\left(\sum_{i=1}^n(1+\theta_j)\tilde{I}_{ij}(X_i)\right)
                -\rho^{Re}_j\left(\sum_{i=1}^n(1+\theta_j)I_{ij}^*(X_i)\right)
                +\frac{\varepsilon}{m}
                \quad\forall j\in\cM\,.
        \end{align*}
    
    \noindent It is straightforward to verify that
        \begin{align*}
            \sum_{i=1}^na_i&=
                \sum_{i=1}^n\sum_{j=1}^m\pi_{ij}^*
                +\sum_{i=1}^n\rho^{In}_i\left(X_i-\sum_{j=1}^mI_{ij}^*(X_i)\right)
                -\sum_{i=1}^n\rho^{In}_i\left(X_i-\sum_{j=1}^m\tilde{I}_{ij}(X_i)\right)\\
            &=\sum_{j=1}^m\sum_{i=1}^n\pi_{ij}^*
                +\sum_{j=1}^m\rho^{Re}_j\left(\sum_{i=1}^n(1+\theta_j)\tilde{I}_{ij}(X_i)\right)
                -\sum_{j=1}^m\rho^{Re}_j\left(\sum_{i=1}^n(1+\theta_j)I^*_{ij}(X_i)\right)
                +\varepsilon\\
            &=\sum_{j=1}^mb_j\,.
        \end{align*}
    Now consider the following system of linear equations with variables $\{\pi_{ij}\}_{i\in\cN,j\in\cM}$:
        \begin{align*}
            \sum_{j=1}^m\pi_{ij}&=a_i\quad\forall i\in\cN\,,\\
            \sum_{i=1}^n\pi_{ij}&=b_j\quad\forall j\in\cM\,.
        \end{align*}
    This is a system of $n+m-1$ linearly independent equations and $n\times m$ variables.
    Since $\sum_{i=1}^na_i=\sum_{j=1}^mb_j$ and there at least as many variables as equations, there exists a solution $\{\tilde{\pi}_{ij}\}_{i\in\cN,j\in\cM}$ to this linear system.
    We now show that the allocation $\allocationtilde$ is a Pareto improvement over $\allocationstar$, providing the desired contradiction.
    Firstly, for each $i\in\cN$, we have
        \begin{align*}
            \rho^{In}_i&\left(X_i-\sum_{j=1}^m\tilde{I}_{ij}(X_i)+\sum_{j=1}^m\tilde{\pi}_{ij}\right)\\
                &=\rho^{In}_i\left(X_i-\sum_{j=1}^m\tilde{I}_{ij}(X_i)\right)
                    +a_i\\
                &=\rho^{In}_i\left(X_i-\sum_{j=1}^m\tilde{I}_{ij}(X_i)\right)
                    +\sum_{j=1}^m\pi_{ij}^*
                    +\rho^{In}_i\left(X_i-\sum_{j=1}^mI_{ij}^*(X_i)\right)
                    -\rho^{In}_i\left(X_i-\sum_{j=1}^m\tilde{I}_{ij}(X_i)\right)\\
                &=\rho^{In}_i\left(X_i-\sum_{j=1}^mI_{ij}^*(X_i)+\sum_{j=1}^m\pi_{ij}^*\right)\,.
        \end{align*}
    Moreover, for each $j\in\cM$, we have
        \begin{align*}
            \rho^{Re}_j\left(\sum_{i=1}^n\big((1+\theta_j)\tilde{I}_{ij}(X_i)-\tilde{\pi}_{ij}\big)\right)
            &=\rho^{Re}_j\left(\sum_{i=1}^n\big((1+\theta_j)\tilde{I}_{ij}(X_i)\big)\right)
                    -b_j\\
                &=\rho^{Re}_j\left(\sum_{i=1}^n\big((1+\theta_j)\tilde{I}_{ij}(X_i)\big)\right)
                    -\sum_{i=1}^n\pi_{ij}^*
                    -\rho^{Re}_j\left(\sum_{i=1}^n(1+\theta_j)\tilde{I}_{ij}(X_i)\right)\\
                &\quad\qquad+\rho^{Re}_j\left(\sum_{i=1}^n(1+\theta_j)I_{ij}^*(X_i)\right)
                    -\frac{\varepsilon}{m}\\
                &<\rho^{Re}_j\left(\sum_{i=1}^n(1+\theta_j)I_{ij}^*(X_i)\right)
                    -\sum_{i=1}^n\pi_{ij}^*\\
                &=\rho^{Re}_j\left(\sum_{i=1}^n\big((1+\theta_j) I^*_{ij}(X_i)-\pi^*_{ij}\big)\right)\,,
        \end{align*}
    
    \noindent where the inequality is strict since $\varepsilon>0$.
    It then follows that $\allocationstar$ is not Pareto optimal, a contradiction.
    Hence, if $\allocationstar$ is PO, then both \ref{lem:po_minimize_risk_1} and \ref{lem:po_minimize_risk_2} must hold.
\end{proof}

To complete the proof of Proposition \ref{prop:optimal_indemnities}, it suffices to show that an indemnity structure is a solution to \eqref{infPO} if and only if it satisfies condition \ref{prop:optimal_indemnities_1} of Proposition \ref{prop:optimal_indemnities}.
To this end, note that since either $\rho_j^{Re}$ is an expectation or the risks $X_1,\ldots,X_n$ are comonotone, we have that for each $j\in\cM$,
    \[
       \rho^{Re}_j\left(\sum_{i=1}^n(1+\theta_j)I_{ij}(X_i)\right)=\sum_{i=1}^n\rho^{Re}_j\left((1+\theta_j)I_{ij}(X_i)\right)\,.
    \]
Hence, the problem \eqref{infPO} is equivalent to
    \begin{equation}\label{infPO_simplified}
        \inf_{\{I_{ij}\}_{i\in\cN,j\in\cM}\in\cI}\left\{
           \sum_{i=1}^n\sum_{j=1}^m\rho^{Re}_j\left((1+\theta_j)I_{ij}(X_i)\right)
            -\sum_{i=1}^n\sum_{j=1}^m\rho^{In}_i \left( I_{ij}(X_i)\right)\right\}\,,
    \end{equation}
where the constant terms $\rho_i^{In}(X_i)$ have been omitted.
For each $i\in\cN$, define the set
    \[
        A_i:=\left\{\sum_{j=1}^m\rho^{Re}_j\left((1+\theta_j)I_{ij}(X_i)\right)
            -\sum_{j=1}^m\rho^{In}_i\left(I_{ij}(X_i)\right)
            \,\,\middle|\,\,
            \{I_{ij}\}_{j\in\cM}\in\cI_i\right\}\,.
    \]
Since $\cI=\prod_{i=1}^n\cI_i$, the optimization problem \eqref{infPO_simplified} simplifies to
    \begin{align*}
        \inf\left(\sum_{i=1}^nA_i\right)=\sum_{i=1}^n\inf A_i\,,
    \end{align*}
since the infimum commutes with the Minkowski sum.
The above decomposes the problem \eqref{infPO_simplified} into $n$ different minimization problems.
That is, $\{I_{ij}^*\}_{i\in\cN,j\in\cM}$ solves \eqref{infPO_simplified} if and only if for each $i\in\cN$, the indemnity structure $\{I_{ij}^*\}_{j\in\cM}\in\cI_i$ solves
    \begin{equation}\label{infPO_simplified_2}
        \inf_{\{I_{ij}\}_{j\in\cM}\in\cI_i}\left\{
           \sum_{j=1}^m\rho^{Re}_j\left((1+\theta_j)I_{ij}(X_i)\right)
            -\sum_{j=1}^m\rho^{In}_i \left( I_{ij}(X_i)\right)\right\}\,.
    \end{equation}
Note that for each insurer $i\in\cN$, problem \eqref{infPO_simplified_2} is precisely the problem of minimizing the aggregate risk of only those contracts involving insurer $i$.
Hence, problem \eqref{infPO_simplified_2} is related to the characterization of Pareto optima in a market with only one insurer.
The result of Proposition \ref{prop:optimal_indemnities} now follows through an application of \cite[Proposition 4.10]{zhu2023equilibria} to each insurer $i\in\cN$.\qed

\bigskip
\section{Details for Assumption \ref{as:sum_1-lipschitz}}
\label{apdx:sum_1-lipschitz}

Note that restricting the feasible collection of indemnities to the set $\mathcal{I}$ ensures that $\sum_{j=1}^m I_{ij}(X_i)$ is also monotone and 1-Lipschitz.
A justification for this assumption is that when an insurer $i$ purchases coverage from multiple reinsurers, the same risk cannot be covered twice.
After a portion of their risk is covered by one reinsurer, only the retained risk is eligible for further coverage -- hence, reinsurance is purchased incrementally.
Suppose that insurer $i$ purchases coverage of an amount $I_{i1}(X_i)$ from reinsurer $1$.
Their retained risk after this transaction is $X_i-I_{i1}(X_i)$.
If insurance is then purchased from reinsurer $2$, then this contract must be of the form $I_{i2}(X_i-I_{i1}(X_i))$.
The following result formalizes this process and shows that it is possible to represent these contracts using the set of admissible indemnity structures $\cI$.

\medskip

\begin{proposition}
    Fix $i\in\cN$, and suppose that for each $j\in\cM$, the functions $I_{ij}$ are non-decreasing and $1$-Lipschitz.
    For each $j\in\cM$, let $R_j$ be the retained risk after purchasing the contracts $I_{i1},\ldots,I_{ij}$.
    That is, define recursively
        \begin{align*}
            R_0&:=X_i\,,\\
            R_j&:=R_{j-1}-I_{ij}(R_{j-1})\,.
        \end{align*}
    Then there exist functions $\{\tilde{I}_{ij}\}_{j\in\cM}\in\cI_i$ such that
        \[
            \tilde{I}_{ij}(X_i)=I_{ij}(R_{j-1}), \ \ \forall \, j\in\cM.
        \]
\end{proposition}

\begin{proof}
    We proceed by induction on $m$. The result clearly holds for $m=1$. Suppose for the sake of induction that the result holds for $m$. 
    Then it is straightforward to verify that, for all $j\in\cM$,
        \[
            R_j=X_i-\sum_{k=1}^j\tilde{I}_{ik}(X_i).
        \]
    Define the function $\tilde{I}_{i,m+1}:[0,M]\to[0,M]$ by $\tilde{I}_{i,m+1}(x) =  I_{i,m+1}\left(x-\sum_{j=1}^m\tilde{I}_{ij}(x)\right)$. It then follows that
        \begin{align*}
            \tilde{I}_{i,m+1}(X_i)&=I_{i,m+1}\left(X_i-\sum_{j=1}^m\tilde{I}_{ij}(X_i)\right)
            =I_{i,m+1}\left(R_m\right)\,.
        \end{align*}
    It remains to show that both $\tilde{I}_{i,m+1}$ and $\sum_{j=1}^{m+1}\tilde{I}_{ij}$ are non-decreasing and 1-Lipschitz.
    First, since $\sum_{j=1}^m\tilde{I}_{ij}$ is non-decreasing and 1-Lipschitz by the inductive assumption, then so is the function
        \[
            x\mapsto x-\sum_{j=1}^m\tilde{I}_{ij}(x)\,.
        \]
    Therefore $\tilde{I}_{i,m+1}$ is also non-decreasing and 1-Lipschitz as a composition of non-decreasing 1-Lipschitz functions, which implies that $\sum_{j=1}^{m+1}\tilde{I}_{ij}$ is non-decreasing.
    To show that $\sum_{j=1}^{m+1}\tilde{I}_{ij}$ is 1-Lipschitz, note that we can write
        \begin{align*}
            \sum_{j=1}^{m+1}\tilde{I}_{ij}(x)
                &=\tilde{I}_{i,m+1}(x)+\sum_{j=1}^m\tilde{I}_{ij}(x)
            =I_{i,m+1}\left(x-\sum_{j=1}^m\tilde{I}_{ij}(x)\right)+\sum_{j=1}^m\tilde{I}_{ij}(x)
            =f(x-g(x))+g(x)\,,
        \end{align*}
    
    \noindent where
        \[
            f:=I_{i,m+1}\quad\mbox{and}\quad g:=\sum_{j=1}^m\tilde{I}_{ij}\,.
        \]
    
    \noindent Note that both $f$ and $g$ are non-decreasing 1-Lipschitz functions.
    For $x>y$, we have
        \begin{align*}
            \frac{f(x-g(x))+g(x)-\big(f(y-g(y))+g(y)\big)}{x-y}
            &=\frac{f(x-g(x))-f(y-g(y)}{x-y}+\frac{g(x)-g(y)}{x-y}\\
            &\le\frac{x-g(x)-(y-g(y))}{x-y}+\frac{g(x)-g(y)}{x-y}
            =1\,,
        \end{align*}
    which shows that $\sum_{j=1}^{m+1}\tilde{I}_{ij}$ is 1-Lipschitz.
\end{proof}

For the sake of notation, it is simpler to consider the contracts $\tilde{I}_{ij}$, since these are functions of $X_i$ rather than functions of retained risk $R_{j-1}$.
We therefore impose Assumption \ref{as:sum_1-lipschitz} in this paper to represent the above process of incremental reinsurance.

\vspace{0.8cm}

\bibliographystyle{ecta}

\vspace{0.4cm}

\end{document}